\documentclass[10pt,twocolumn]{article}
\usepackage[letterpaper,margin=0.72in,columnsep=0.24in]{geometry}
\usepackage[T1]{fontenc}
\usepackage{lmodern}
\usepackage{microtype}
\usepackage{amsmath,amssymb,amsfonts,mathtools,bm,amsthm}
\usepackage{booktabs,array,longtable,multirow,xcolor,graphicx,float,enumitem,placeins}
\usepackage{cuted,dblfloatfix}
\usepackage{tikz,quantikz}
\usetikzlibrary{arrows.meta,positioning,fit,calc}
\usepackage{xurl}
\ifdefined\supplementonly
\usepackage{xr}
\fi
\usepackage{fancyhdr}
\usepackage[hidelinks]{hyperref}
\usepackage{cite}
\usepackage[nameinlink]{cleveref}
\ifdefined\supplementonly
\fi
\setlist{nosep,leftmargin=*}
\numberwithin{equation}{section}
\newtheorem{theorem}{Theorem}[section]
\newtheorem{proposition}[theorem]{Proposition}

\newtheorem{corollary}[theorem]{Corollary}
\theoremstyle{definition}\newtheorem{definition}[theorem]{Definition}

\theoremstyle{remark}\newtheorem{remark}[theorem]{Remark}
\newtheorem{example}[theorem]{Example}
\newcommand{\R}{\mathbb R}\newcommand{\C}{\mathbb C}
\newcommand{\E}{\mathbb E}
\newcommand{\Cov}{\operatorname{Cov}}\newcommand{\Tr}{\operatorname{Tr}}
\newcommand{\TV}{\operatorname{TV}}

\definecolor{designblue}{RGB}{63,81,181}
\definecolor{responseorange}{RGB}{230,126,34}
\definecolor{certgreen}{RGB}{39,145,91}

\title{\vspace{-1.8em}\textbf{Completing or Refusing Low-Dimensional Records of Structured Quantum Circuits:\\
Measurement Loss, Compression Loss, and Hardware Drift}}
\author{Gunhee Cho$^{1}$\thanks{Corresponding author. Email: \href{mailto:wvx17@txstate.edu}{wvx17@txstate.edu}.} and
Juhee Lee$^{2}$\thanks{Second author. Email: \href{mailto:jhlee@sdt.inc}{jhlee@sdt.inc}.}\\
\small $^{1}$Texas State University, San Marcos, USA\\
\small $^{2}$SDT Inc., Seoul, Republic of Korea}
\date{}

\begin{document}
\ifdefined\supplementonly
\title{Supplementary Material for\\
Completing or Refusing Low-Dimensional Records of Structured Quantum Circuits}
\maketitle
\thispagestyle{fancy}
\else
\maketitle
\thispagestyle{fancy}
\begin{abstract}
Structured quantum circuits are often represented by low-dimensional records
such as occupation numbers or constraint feasibility rather than by their full
outcome distributions. Such a record loses control-dependent information when
it merges outcomes whose probabilities respond differently to the circuit
parameters. We assess this loss without assuming a parametric hardware-noise
model or a low-dimensional statistical family. Quantum Fisher information
bounds the sensitivity available before measurement, while the Fisher metrics
of admissible measurements and their recorded statistics describe the
sensitivity attainable on the device. For finite control changes, a Hellinger
residual measures the response removed by the record; locally, its quadratic
term is determined by the conditional covariance of the full-outcome score.
Simultaneous confidence bounds then support one of three decisions: approve the
record, refuse it, or defer the decision.

We prove a finite-library completion theorem for a declared collection of
state or distribution responses. A procedure using executable augmentations
terminates either with a record that preserves every declared response or with
proof that no augmentation in the library removes the remaining loss. For an
analytic control germ, we use
integral closures to characterize preservation along every analytic control
arc and finitely many Rees valuations to detect failure. Our contribution is to
define the quantum state--measurement--record response chain, connect its
all-arc criterion to an executable completion-or-refusal procedure, state the
corresponding local kernel criterion for attainable Fisher metrics, and retain
the same refusal rule under finite-sample uncertainty.

A three-qubit calculation first distinguishes loss caused by measurement from
loss caused by the record. Prospectively fixed IBM experiments then test the
predicted decisions on Kingston and Marrakesh. In a four-qubit
fixed-particle-number family, the procedure refuses Hamming weight and
one-sided occupation but approves a number record that resolves the occupied
positions. In a four-qubit GHZ family, $Z$-basis invariance and the odd
$Y$-parity response replicate across devices and blocks, whereas the
preregistered $X$-even conjunction fails and remains a negative control. In
the two-epoch IQM Garnet data, the median record-level
drift is $0.0797$ times the full-distribution drift, but a significant residual
remains in 35 of 36 settings. On these circuits, the predeclared rules detect
loss, approve a record only within its stated tolerance, and otherwise refuse
or defer it; the experiments do not
establish universal compression performance or device quantum Fisher
information.
\end{abstract}

\tableofcontents

\section{Introduction}
Structured quantum circuits restrict their ideal outputs to physically
specified subsets of the outcome space. Examples include constraint-preserving
optimization and fixed-particle-number computation
~\cite{Hadfield2019,Monbroussou2023,Farias2024}. This support constraint,
however, does not establish that a low-dimensional record is adequate for
hardware data. Outcomes assigned the same feasibility label, Hamming weight,
or occupation value may respond differently to circuit controls, while the
record itself may also respond to calibration and readout changes.
Coarse-graining can therefore suppress a target response, transmit nuisance
drift as an apparent physical response, or force a choice among records that
all leave substantial residual information. In the last case, failure to find
an adequate representation is misreported as evidence that the selected
representation is valid.

Against this background, we ask whether a proposed record preserves the
measured response without first specifying a hardware-noise model or a
low-dimensional statistical family. Answering this question requires
separating loss caused by measurement from loss caused by subsequent
coarse-graining, and distinguishing the target response from nuisance hardware
drift. A record is refused only when a certified residual exceeds its declared
loss budget; a large full-distribution response may be scientifically intended
and does not by itself make the record inadequate. If finite data support
neither approval nor refusal, the decision is deferred. An independently
designed follow-up may resolve the new comparison, but it does not replace the
original decision.

The three-qubit circuit in Section~2 makes this distinction explicit. The
one-site record $T_1=n_1$ merges $\lvert100\rangle$ and
$\lvert001\rangle$; it is therefore independent of the control $b$, even
though the full computational-basis distribution changes with $b$. The
augmented record $T_2=(L,n_1,n_2)$ separates all three allowed single-particle
outcomes and exactly recovers the ideal computational-basis histogram. Its use
on hardware nevertheless requires a residual check because it does not
distinguish every outcome outside the ideal support. A relative phase $\phi$
provides the complementary case: it changes the quantum state while leaving
all computational-basis probabilities unchanged, so no finer statistic of the
same counts can recover it. Recovering $\phi$ requires an interference
measurement. Thus the $b$ direction exhibits compression loss, whereas the
$\phi$ direction exhibits measurement loss, and the two failures require
different augmentations.

In classical statistics, sufficiency characterizes when a statistic preserves
all parameter information, and the Fisher metric is monotone under statistical
transformations~\cite{AyJostLeSchwachhofer2015}. Quantitative notions of almost
sufficiency instead bound the Fisher loss by comparing the original and induced
metrics~\cite{YamaguchiNozawa2024}. In quantum models, the Braunstein--Caves
relation identifies the SLD quantum Fisher information as the optimal local
Fisher information over measurements, while the theory of monotone quantum
metrics describes contraction under quantum channels
~\cite{BraunsteinCaves1994,PetzGhinea2011}. For a broad class of such metrics,
preservation under a channel is equivalent to recoverability, but this
equivalence does not hold for the SLD metric in general~\cite{Gao2023}.
Comparison theory for quantum statistical experiments gives a stronger,
decision-theoretic criterion: one model can simulate another for all decision
problems when the required statistical transformation exists
~\cite{Buscemi2012}. These results specify important notions of information
preservation, but they do not by themselves determine whether a fixed hardware
measurement followed by a chosen record preserves a declared response.

Several adjacent lines of work optimize the measurement rather than audit a
fixed record.  Fidelity-optimal measurements identify POVMs that attain a
specified distinguishability, while minimal-sufficiency theory removes
statistically redundant outcomes
~\cite{FuchsCaves1995,Kuramochi2015,ChenZhu2025}.  Experiments have extracted
Fisher information from non-Gaussian spin distributions
~\cite{Strobel2014}; finite-resolution theory shows that coarse detectors can
miss large quantum Fisher information and that preprocessing can restore part
of it~\cite{Frowis2016}.  Coarse-grained thermometry instead optimizes a small
number of energy bins for retaining temperature Fisher information
~\cite{Hovhannisyan2021}.  Preprocessing-optimized Fisher information chooses
controls before a noisy measurement, whereas randomized designs provide
near-optimal compromises for multiparameter models
~\cite{ZhouMichalakisGefen2023,ZhouChen2026}.  These approaches choose or
optimize a measurement for a specified state family and task.  Our problem
begins after the device, admissible measurement library, and interpretable
records have been declared: it asks whether the chosen record should be
approved, augmented, or refused.

Hellinger-based circuit
reproducibility studies~\cite{DasguptaHumble2022,DahlhauserHumble2024,
Hashim2025,Lorenz2025} do not provide this stage-specific decision either.
Local Fisher sensitivity need not determine finite-shot performance
~\cite{MeyerKhatriFranca2025}.  We therefore base finite decisions on joint
distributional confidence regions rather than on a QFIM alone.

We compare the information available at three successive levels: the quantum
state, the permitted measurement, and the retained record.  The first
transition identifies information that the available measurement cannot
access; the second identifies information discarded only by the chosen record.
Their local forms are measured by differences of Fisher metrics, while finite
changes are compared through Hellinger distance.  Simultaneous confidence
regions then support one of three decisions: approve the record, refuse it, or
defer the decision when the data are insufficient.  Claims that must hold along
every analytic control curve require the additional response-ideal audit.  Any
algebraic witness obtained from that audit is treated as a candidate for a
follow-up experiment only after its physical executability has been
established.

The algebraic criteria used here come from classical integral-closure and
Rees-valuation theory: integral closure is characterized by analytic arcs, and
its failure can be detected by finitely many Rees valuations
~\cite{LejeuneTeissier2008,HunekeSwanson2006}.  We apply these criteria to an
ordered family of quantum state, measurement, and record response ideals so
that each failed inclusion identifies a distinct source of information loss.
The accessible-completion theorem is the result introduced here. Relative to
a declared finite library of executable measurements or records, it terminates
either with an augmentation that preserves every declared response direction
or with proof that no member of the library does so. Its algebraic branch turns Rees witnesses into a
terminating augmentation test, while its local branch gives the corresponding
kernel criterion for attainable Fisher metrics.  The conditional-score
identity, data-processing inequalities, and quantum-Fisher covariance formula
are used as established ingredients.

The four main theorems address the same completeness-or-refusal question at
different levels. Theorem~\ref{en:thm:two-level-completeness}
characterizes preservation along every analytic control curve.
Theorem~\ref{en:thm:finite-rees-audit} converts failure of that criterion into
a finite algebraic witness, and Theorem~\ref{en:thm:accessible-completion}
determines whether the declared augmentation library can complete the record.
Theorem~\ref{en:thm:three-way} then gives the approve--refuse--defer decision
under finite-sample uncertainty.  The first three results concern model-level
response across analytic control curves, whereas the fourth concerns observed
finite comparisons.  The response-ideal analysis and the Hellinger audit are
therefore the global and finite-observation levels of the same decision
problem, not independent objectives.

Most experimental users need only the five-step finite audit in Section~3.
They retain the full outcome counts, declare the candidate record and loss
budget, compare the full and recorded responses, construct simultaneous
uncertainty bounds, and return approve, refuse, or defer for the specified
control comparisons.  Steps~6--7 are reserved for claims of
first-response-order preservation along every analytic control curve or for
designing a follow-up experiment from an algebraic failure witness.  A single
comparison whose residual is certified above the loss budget is already
sufficient to refuse the record on that comparison; no integral-closure
calculation is required.

The hardware studies test the proposed decision rules rather than the
response-ideal equalities themselves.  The prospectively fixed three-qubit
Kingston study tests record loss, measurement loss, and the sequence from defer
to an independently designed follow-up.  A protocol-preserving replication
applies the same circuit family, control points, randomization, shot counts,
and decision rules on Marrakesh and in a new Kingston block, with only
calibration-dependent physical layouts reselected.  A separate four-qubit
experiment tests the full decision ladder on GHZ phase circuits and
fixed-particle-number exchange circuits.  It includes a failed preregistered
symmetry check, an augmentation that improves the record but remains
insufficient, and a resolved number record that is approved.

The two-epoch IQM Garnet study examines residual loss and drift attenuation on a
larger outcome space.  Across these tested circuits, the predeclared rules
detected finite-direction loss and returned augmentation, approval, refusal,
or deferment according to the stated tolerances.  The studies do not establish
equality of response-ideal integral closures, device quantum Fisher
information, universal compression performance, or hardware superiority.

Section~2 defines the experimental model, record statistics, Fisher and
Hellinger losses, and the response ideals needed for all-direction claims, and
evaluates these objects in the common three-qubit example.  Section~3 develops
the four main theorems and the seven-step completion-or-refusal procedure.
Section~4 tests the finite-direction decision rules in the IBM and IQM
studies.  Section~5 positions the results relative to prior work and states
their mathematical and experimental limits.  The separately compiled
supplement contains the longer proofs, protocols, and reproducibility records.

Figure~\ref{en:fig:decision-map} summarizes the two analysis routes used
throughout the paper.  The upper route is the default finite-data audit: it
starts from full outcome counts, evaluates a proposed record, and returns
approve, refuse, or defer.  The lower route is used only for claims that
quantify over every analytic control curve; it applies response ideals and the
completion theorem to a declared library of executable augmentations.  The two
routes concern the same response target but support claims of different
strength.  The algebraic route does not replace finite-sample inference, and
the finite audit alone does not certify preservation along every analytic
control curve.
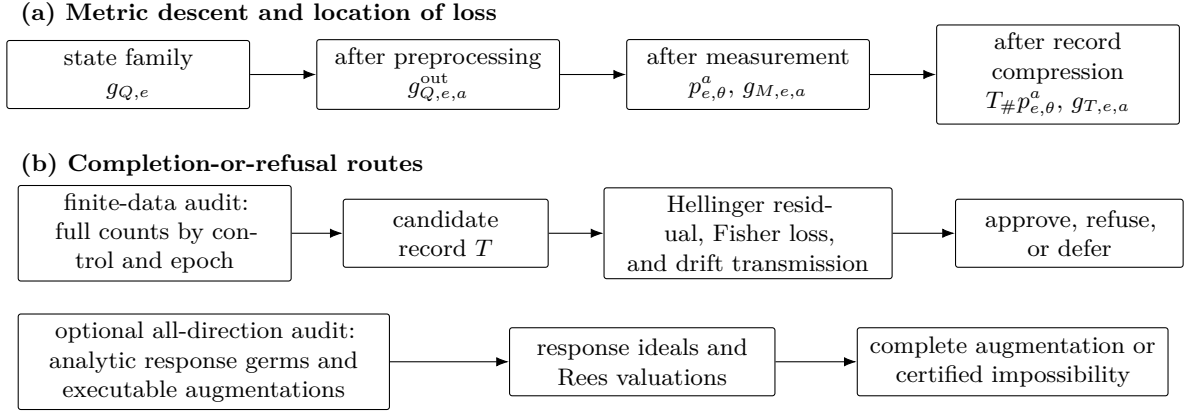
\begin{figure*}[t]
\centering
\begin{tikzpicture}[>=Latex,
 every node/.style={font=\small},
 stage/.style={draw,rounded corners=1pt,align=center,inner sep=3pt,
 minimum height=9mm}]
\node[anchor=west,font=\small\bfseries] at (0,3.05)
 {(a) Metric descent and location of loss};
\node[stage,text width=30mm] (q) at (1.55,2.25)
 {state family\\$g_{Q,e}$};
\node[stage,text width=30mm] (qo) at (5.65,2.25)
 {after preprocessing\\$g^{\rm out}_{Q,e,a}$};
\node[stage,text width=30mm] (m) at (9.75,2.25)
 {after measurement\\$p^a_{e,\theta}$, $g_{M,e,a}$};
\node[stage,text width=30mm] (t) at (13.85,2.25)
 {after record compression\\$T_\#p^a_{e,\theta}$, $g_{T,e,a}$};
\draw[->] (q)--(qo); \draw[->] (qo)--(m); \draw[->] (m)--(t);

\node[anchor=west,font=\small\bfseries] at (0,1.05)
 {(b) Completion-or-refusal routes};
\node[stage,text width=34mm] (fc) at (1.9,0.15)
 {finite-data audit:\\full counts by control and epoch};
\node[stage,text width=25mm] (cr) at (5.75,0.15) {candidate record $T$};
\node[stage,text width=36mm] (fr) at (9.75,0.15)
 {Hellinger residual, Fisher loss,\\and drift transmission};
\node[stage,text width=28mm] (fd) at (14.0,0.15)
 {approve, refuse,\\or defer};
\draw[->] (fc)--(cr); \draw[->] (cr)--(fr); \draw[->] (fr)--(fd);

\node[stage,text width=47mm] (ag) at (2.55,-1.55)
 {optional all-direction audit:\\analytic response germs and executable augmentations};
\node[stage,text width=33mm] (ri) at (8.35,-1.55)
 {response ideals and\\Rees valuations};
\node[stage,text width=39mm] (ac) at (13.25,-1.55)
 {complete augmentation or\\certified impossibility};
\draw[->] (ag)--(ri); \draw[->] (ri)--(ac);
\end{tikzpicture}
\caption{Information descent and decision routes.  Panel (a) locates loss
introduced by preprocessing, measurement, and record compression.  Panel (b)
separates the default finite-data audit from the optional algebraic audit for
claims over all analytic control curves.  The latter requires analytic response
germs and a declared library of physically executable augmentations; it is not
inferred from finite counts alone.}
\label{en:fig:decision-map}
\end{figure*}

\section{Setting and the three-qubit running example}
\label{en:sec:mathematical-setting}

\subsection{Objects, maps, and comparison levels}
We introduce the mathematical objects in order of claim strength.  A finite
hardware audit requires the experimental model, a record map, Hellinger
distance, and outcome counts.  A local sensitivity audit additionally requires
scores and Fisher matrices.  Response ideals enter only when completeness is
claimed along every analytic control curve.  Readers concerned only with
finite hardware decisions can therefore carry out Steps~1--5 without the
algebraic machinery reserved for Steps~6--7.

Let $\mathbb T=\R/(2\pi\mathbb Z)$ be the circle of real angles modulo
$2\pi$, with addition induced from $\R$, and fix $d\in\mathbb N_{\geq1}$.  The
global control space is the torus $\Theta=\mathbb T^d$.  To define local
derivatives, choose a reference point $\theta_0\in\Theta$ and a chart
$U_0\subset\Theta$ identified with an open subset of $\R^d$.  All derivatives,
analytic germs, and response ideals are taken in this chart.  This convention
keeps the physical periodicity of the controls separate from the local
analytic statements.

\begin{definition}[One-epoch compiled-noisy measurement experiment]
\label{en:def:epoch-experiment}
Fix $n_q,D\in\mathbb N_{\geq1}$.  Let $V=\{0,\ldots,n_q-1\}$ be the set of
qubit labels and let $\mathcal X=\{0,1\}^{V}$ be the set of
computational-basis bitstrings.  Let $\mathcal H$ be a $D$-dimensional complex
Hilbert space.  Write
\[
\begin{aligned}
 \mathcal L(\mathcal H)&=\{A:\mathcal H\to\mathcal H\text{ linear}\},\\
 \mathcal D(\mathcal H)&=\{\rho\in\mathcal L(\mathcal H):
 \rho\succeq0,\ \Tr\rho=1\}.
\end{aligned}
\]
We also define
\[
 \operatorname{Pos}(\mathcal H)=
 \{A\in\mathcal L(\mathcal H):A\succeq0\},
\]
and let $\operatorname{CPTP}(\mathcal H)$ denote the set of completely
positive, trace-preserving linear maps from $\mathcal L(\mathcal H)$ to
itself.
For a linear map $\Phi:\mathcal L(\mathcal H)\to\mathcal L(\mathcal H)$,
complete positivity means that
$\operatorname{id}_{\mathcal L(\C^k)}\otimes\Phi$ maps
$\operatorname{Pos}(\C^k\otimes\mathcal H)$ into itself for every $k\geq1$.
Let $\mathsf E$ be a nonempty set of epoch labels.  A calibration epoch
$e\in\mathsf E$ is a designated execution block in which the same residual
channel $\mathcal E_e$ and readout POVM $\mathcal M_e$ are assigned to every
control setting executed in that block.  This stationarity is an analysis
assumption, not a physical law; another epoch may have different maps.  The
typed environment for epoch $e$ is
\begin{equation}\label{en:eq:epoch-environment}
 \mathbf E_e=\bigl(\Theta,\mathcal H,\rho_0,
 \mathcal C,\mathcal E_e,\mathcal M_e,\mathcal Y,n\bigr),
\end{equation}
where
\begin{align*}
 &\rho_0\in\mathcal D(\mathcal H),\qquad
 \mathcal C:\Theta\to\operatorname{CPTP}(\mathcal H),\quad
 \theta\mapsto\mathcal C_\theta,\\
 &\mathcal E_e\in\operatorname{CPTP}(\mathcal H),\qquad
 \varnothing\neq\mathcal Y\text{ is finite},\\
 &\mathcal M_e:\mathcal Y\to\operatorname{Pos}(\mathcal H),\quad
 y\mapsto M_{e,y},\qquad
 \sum_{y\in\mathcal Y}M_{e,y}=I_{\mathcal H},\\
 &n\in\mathbb N_{\geq1}.
\end{align*}
Here $\rho_0$ is the prepared initial state, $\mathcal C_\theta$ is the native
implementation produced by the compiler for control $\theta$, and
$\mathcal E_e$ is the residual channel common to the compiled family during
epoch $e$.  The family $\mathcal M_e=\{M_{e,y}:y\in\mathcal Y\}$ is the
readout-inclusive POVM, and $n$ is the number of repetitions at one control
setting.  For computational-basis readout one may take
$\mathcal Y=\mathcal X$; the general notation also permits other finite-outcome
measurements.  In finite dimensions, every
$\Phi\in\operatorname{CPTP}(\mathcal H)$ admits a Kraus representation
\begin{equation}\label{en:eq:kraus-setting}
 \Phi(A)=\sum_{\alpha=1}^{r}K_\alpha A K_\alpha^\dagger,
 \qquad \sum_{\alpha=1}^{r}K_\alpha^\dagger K_\alpha=I_{\mathcal H}.
\end{equation}
for some finite $r$ and $K_\alpha\in\mathcal L(\mathcal H)$.  Dephasing,
amplitude damping, and depolarizing channels may be included in $\mathcal E_e$,
while readout errors are included in $\mathcal M_e$.  Multinomial variation
arises only after measurement and is not a quantum channel.  Control-dependent
compilation effects belong to $\mathcal C_\theta$, and drift across epochs is
represented by changes in $\mathcal E_e$ and $\mathcal M_e$.  This
factorization fixes the stages to be audited; it is an analysis convention
rather than a claim that the microscopic noise decomposition is uniquely
identifiable.

Define the compiled state, post-channel state, and observed outcome law by
\begin{equation}\label{en:eq:epoch-born-law}
 \begin{aligned}
 \sigma_\theta&=\mathcal C_\theta(\rho_0),\\
 \omega_{e,\theta}&=\mathcal E_e(\sigma_\theta),\\
 p_{e,\theta}(y)&=\Tr(M_{e,y}\omega_{e,\theta}).
 \end{aligned}
\end{equation}
The CPTP and POVM conditions imply
$\sigma_\theta,\omega_{e,\theta}\in\mathcal D(\mathcal H)$ and
$p_{e,\theta}\in\Delta(\mathcal Y)$, where
\[
 \begin{aligned}
 \Delta(\mathcal Y)&=\left\{p:\mathcal Y\to[0,1]:
 \sum_{y\in\mathcal Y}p(y)=1\right\},\\
 \mathcal P_e&=\{p_{e,\theta}\in\Delta(\mathcal Y):\theta\in\Theta\}.
 \end{aligned}
\]
The family $\mathcal P_e$ is the observed response family for epoch $e$.
It, rather than a uniquely identified microscopic channel, is the direct
statistical object of the audit.
Fix $(e,\theta)$.  On a probability space
$(\Omega_{e,\theta},\mathcal F_{e,\theta},\Pr_{e,\theta})$, assume that the
shot outcomes $Y_{e,\theta,1},\ldots,Y_{e,\theta,n}$ are conditionally
independent and identically distributed with categorical law $p_{e,\theta}$.
Define
\begin{equation}\label{en:eq:epoch-count-law}
 \begin{aligned}
 N_{e,\theta}(y)&=\sum_{i=1}^{n}\mathbf1\{Y_{e,\theta,i}=y\},\\
 N_{e,\theta}&\sim\operatorname{Multinomial}(n,p_{e,\theta}),\\
 \widehat p_{e,\theta}&=N_{e,\theta}/n.
 \end{aligned}
\end{equation}
Thus $N_{e,\theta}:\Omega_{e,\theta}\to\mathcal N_n(\mathcal Y)$, where
$\mathcal N_n(\mathcal Y)=\{c\in\mathbb N_0^{\mathcal Y}:
\sum_{y\in\mathcal Y}c(y)=n\}$.  For fixed $(e,\theta)$, $p_{e,\theta}$ is
fixed but unknown, whereas the shot outcomes, count vector, and empirical
distribution are random.  This distinction separates variation of the device
response from finite-shot uncertainty.
\end{definition}

The compiler may produce a distinct native implementation $\mathcal C_\theta$
for each control value $\theta$; we therefore do not assume that the compiled
family is generated by a single parameter-independent channel.  The
differential and response-ideal results require
$\theta\mapsto\sigma_\theta$ to be analytic only within the selected local
compilation branch.  If the gate layout or routing changes discontinuously,
the two sides are treated as separate local response families.  In a paired
cross-epoch comparison, the compiled circuits are held fixed while
$\mathcal E_e$ and $\mathcal M_e$ may vary with $e$.  Holding the circuits fixed
controls recompilation effects; it does not remove hardware drift.

\begin{definition}[Metric descent from the quantum model to a hardware record]
\label{en:def:metric-descent}
Let
$\omega_e:\Theta\to\mathcal D(\mathcal H)$,
$\theta\mapsto\omega_{e,\theta}$, be smooth and have constant rank on a
neighborhood $U$ of the point under study.  On $U$, let
$g_{\rm SLD}=4g_{\rm Bures}$ denote the symmetric-logarithmic-derivative quantum
Fisher metric, let $g_{\rm FR}$ denote the Fisher--Rao metric on a finite
probability simplex, and define $g_{Q,e}:=\omega_e^*g_{\rm SLD}$.  At a
rank-changing point, we use only an explicitly specified one-sided limit or
lower-semicontinuous extension; an equality established on a constant-rank
neighborhood is not asserted across that point without a separate argument.
An executable analysis setting $a=(\Phi_{e,a},\mathsf M_{e,a})$ consists of a
$\theta$-independent CPTP preprocessing map $\Phi_{e,a}$ and a POVM
$\mathsf M_{e,a}=\{M^a_{e,y}:y\in\mathcal Y_a\}$ on a nonempty finite outcome
set $\mathcal Y_a$.  Let $\varnothing\ne\mathcal T$ be finite and let
$T:\mathcal Y_a\to\mathcal T$ be a record.  The setting $a$ and record $T$
induce
\[
 \begin{aligned}
 \omega_{e,\theta}&\longmapsto\Phi_{e,a}(\omega_{e,\theta})\\
 &\longmapsto p^a_{e,\theta}(y)
 :=\Tr\!\left[M^a_{e,y}\Phi_{e,a}(\omega_{e,\theta})\right]\\
 &\longmapsto T_\#p^a_{e,\theta}.
 \end{aligned}
\]
Writing $p_e^a:\Theta\to\Delta(\mathcal Y_a)$ for
$\theta\mapsto p^a_{e,\theta}$, define
\begin{align}
 g^{\rm out}_{Q,e,a}&:=(\Phi_{e,a}\circ\omega_e)^*g_{\rm SLD},\\
 g_{M,e,a}&:=(p_e^a)^*g_{\rm FR},\\
 g_{T,e,a}&:=(T_\#\circ p_e^a)^*g_{\rm FR}.
 \label{en:eq:metric-descent}
\end{align}
The four metrics locate sensitivity loss along the operational chain.  The
metric $g_{Q,e}$ describes the quantum sensitivity of the state family before
analysis preprocessing; $g^{\rm out}_{Q,e,a}$ is the sensitivity remaining
after $\Phi_{e,a}$; $g_{M,e,a}$ is the sensitivity accessible through the POVM
in setting $a$; and $g_{T,e,a}$ is the sensitivity retained after recording
only $T$.  Consequently,
\[
 g_{Q,e}-g^{\rm out}_{Q,e,a},\qquad
 g^{\rm out}_{Q,e,a}-g_{M,e,a},\qquad
 g_{M,e,a}-g_{T,e,a}
\]
attribute local loss, respectively, to preprocessing, measurement, and record
compression.  A numerical QFIM alone does not determine the downstream
classical metric: distinct quantum statistical models can have the same QFIM
while producing different Fisher information under a fixed POVM.  The state
family, its tangent operators, and the admissible quantum-to-classical maps
must therefore remain part of the specification.
\end{definition}

\begin{definition}[Resource-constrained hardware-attainable geometry]
\label{en:def:attainable-geometry}
Let $\varnothing\ne\mathfrak A_e(B)$ be the set of analysis settings executable
in epoch $e$ under resource budget $B$.  Let $\mathfrak D_e(B)$ be the allowed
finitely supported probability distributions
$\lambda=(\lambda_a)_{a\in\mathfrak A_e(B)}$ for randomized experiments in
which the setting label is retained with each outcome.  At $\theta$, define
\begin{equation}
 \mathcal G_{\rm att}(e,\theta;B)=
 \overline{\left\{\sum_{a\in\mathfrak A_e(B)}\lambda_a g_{M,e,a}(\theta):
 \lambda\in\mathfrak D_e(B)\right\}}.
 \label{en:eq:attainable-fibre}
\end{equation}
Here the closure is taken in the finite-dimensional space
$\operatorname{Sym}_+^2(T_\theta^*\Theta)$ of positive-semidefinite quadratic
forms.  For $v\in T_\theta\Theta$, define the directional support function
\[
 h_{\rm att}(e,\theta;v)=
 \sup_{g\in\mathcal G_{\rm att}(e,\theta;B)}g(v,v).
\]
Optimal settings for different tangent directions need not be jointly
executable within the same budget.  Accordingly, $\mathcal G_{\rm att}$ cannot
in general be represented by a single ``hardware FIM.''  A finite library of
settings that have been explicitly enumerated and verified to be executable
determines an inner subset of the full attainable fibre; finite-shot estimates
of that subset must still be accompanied by their statistical uncertainty.
To retain the loss introduced by the final measurement, define the paired
attainable fibre
\begin{equation}
 \mathfrak G_{\rm att}^{Q\to M}(e,\theta;B)
 =\overline{\left\{\begin{array}{l}
 \left(\sum_a\lambda_a g^{\rm out}_{Q,e,a},
 \sum_a\lambda_a g_{M,e,a}\right):\\[-2pt]
 \lambda\in\mathfrak D_e(B)
 \end{array}\right\}}.
 \label{en:eq:paired-attainable-fibre}
\end{equation}
where the closure is taken in the corresponding product space of
positive-semidefinite quadratic forms.  Projection onto the second coordinate
gives $\mathcal G_{\rm att}$.  Retaining both coordinates isolates the loss
introduced by the final POVM; comparison of the first coordinate with
$g_{Q,e}$ separately identifies loss introduced by preprocessing.
\end{definition}

\begin{proposition}[Quantum--channel--measurement--record monotonicity]
\label{en:prop:metric-descent}
For every admissible setting $a$ and record $T$,
\begin{equation}
 g_{T,e,a}\preceq g_{M,e,a}\preceq
 g^{\rm out}_{Q,e,a}\preceq g_{Q,e}.
 \label{en:eq:metric-monotone-chain}
\end{equation}
The three positive-semidefinite differences
\[
 g_{M,e,a}-g_{T,e,a},\qquad
 g^{\rm out}_{Q,e,a}-g_{M,e,a},\qquad
 g_{Q,e}-g^{\rm out}_{Q,e,a}
\]
are the local losses introduced, respectively, by record compression, the
final POVM, and preprocessing.  The first difference is the within-cell
conditional-score covariance; the other two follow from monotonicity of the
Fisher--Rao and Bures metrics.  At a regular point, a local record audit
compares $g_{T,e,a}$ with $g_{M,e,a}$; finite-separation approval or refusal
instead uses the Hellinger procedure developed below.

Fix a baseline setting $a_0$ and a tangent direction $v$.  Measurement
augmentation is operationally justified only if an admissible design attains
sensitivity greater than $g_{M,e,a_0}(v,v)$.  If
\[
 \begin{aligned}
 g_{M,e,a_0}(v,v)&<g_{Q,e}(v,v),\\
 h_{\rm att}(e,\theta;v)&=g_{M,e,a_0}(v,v),
 \end{aligned}
\]
then the remaining gap is quantum sensitivity inaccessible under the declared
device budget.  It cannot be recovered merely by refining the record.
\end{proposition}
\begin{proof}
The first inequality is Fisher--Rao data processing under the Markov kernel
induced by $T$.  The second is the Braunstein--Caves bound for the specified
POVM, and the third is monotonicity of the SLD/Bures metric under the
parameter-independent CPTP map $\Phi_{e,a}$.  Evaluating each quadratic form
on an arbitrary tangent vector gives the positive-semidefinite order.
\end{proof}

\begin{proposition}[Label-retaining randomization at finite separation]
\label{en:prop:labelled-randomization}
Let $A\sim\lambda$ be independent of $\theta$, and retain the setting label $A$
with each outcome.  The joint law on
$\bigsqcup_{a\in\operatorname{supp}\lambda}\{a\}\times\mathcal Y_a$ is
$P^\lambda_{e,\theta}(a,y):=\lambda_a p^a_{e,\theta}(y)$.  For each $a$, let
$T_a:\mathcal Y_a\to\mathcal T_a$ be a record, define
$g_{T_a,e,a}:=(T_{a\#}\circ p_e^a)^*g_{\rm FR}$, and set
$R(a,y):=(a,T_a(y))$.  Then, for any $\theta,\theta'$,
\begin{equation}
 \begin{aligned}
 H^2(P^\lambda_{e,\theta},P^\lambda_{e,\theta'})
 &=\sum_a\lambda_a H^2(p^a_{e,\theta},p^a_{e,\theta'}),\\
 H^2(R_\#P^\lambda_{e,\theta},R_\#P^\lambda_{e,\theta'})
 &\le H^2(P^\lambda_{e,\theta},P^\lambda_{e,\theta'}).
 \end{aligned}
 \label{en:eq:labelled-hellinger}
\end{equation}
At a regular point, the per-shot Fisher metrics before and after recording are,
respectively, $\sum_a\lambda_a g_{M,e,a}$ and
$\sum_a\lambda_a g_{T_a,e,a}$.  If the setting label is discarded or otherwise
postprocessed, Hellinger and Fisher data-processing inequalities remain valid,
but the displayed equalities need not hold.
\end{proposition}
\begin{proof}
The Bhattacharyya coefficient of the labelled laws is
$\sum_a\lambda_a\sum_{y\in\mathcal Y_a}
\sqrt{p^a_{e,\theta}(y)p^a_{e,\theta'}(y)}$, which gives the first identity.
The second is Hellinger data processing under $R$.  Because $\lambda$ is
parameter independent, the labelled score conditional on $A=a$ is the score
of $p^a_{e,\theta}$; averaging its second moment before and after applying
$T_a$ proves both Fisher formulas.
\end{proof}

\begin{definition}[Deterministic record and allowed-set conditioning]
\label{en:def:quotient-setting}
\label{en:def:leakage-boundary}
Let $\varnothing\ne\mathcal T$ be finite.  A deterministic record is a map
$T:\mathcal Y\to\mathcal T$.  It induces the fiber partition
$\mathcal Y=\bigsqcup_{t\in T(\mathcal Y)}T^{-1}(t)$ and the pushforward
\begin{equation}\label{en:eq:pushforward-setting}
 (T_\#p_{e,\theta})(t)=p_{e,\theta}^T(t)
 =\sum_{y\in T^{-1}(t)}p_{e,\theta}(y).
\end{equation}
Thus compression replaces every fiber by its total probability.  Let
$\varnothing\ne S_{\mathcal Y}\subseteq\mathcal Y$ be a declared allowed-outcome
set.  For computational-basis readout, $\mathcal Y=\mathcal X$ and
$S_{\mathcal Y}=S$.  Define
\begin{equation}\label{en:eq:support-leakage-setting}
 \begin{aligned}
 L(y)&=\mathbf1\{y\notin S_{\mathcal Y}\},\\
 \ell_{e,\theta}&=p_{e,\theta}(S_{\mathcal Y}^c),\\
 u_{e,\theta}&=1-\ell_{e,\theta}.
 \end{aligned}
\end{equation}
If $u_{e,\theta}>0$, define the conditional allowed law by
$q_{e,\theta}(y)=p_{e,\theta}(y)/u_{e,\theta}$ for
$y\in S_{\mathcal Y}$.  The quantity $\ell_{e,\theta}$ is the observed mass
outside the declared allowed set, called \emph{support leakage} below.  It is
not population leakage outside the transmon computational subspace.  All four
quantities are functionals of the observed law and require no microscopic
noise model.
\end{definition}

\begin{definition}[Finite-distance record audit]
\label{en:def:local-response-setting}
Let $\varnothing\ne\mathcal Z$ be finite and let
$p,q\in\Delta(\mathcal Z)$.  Define
\begin{equation}\label{en:eq:metrics-setting}
 \begin{aligned}
 H^2(p,q)&=1-\sum_{z\in\mathcal Z}\sqrt{p(z)q(z)},\\
 \TV(p,q)&=\frac12\sum_{z\in\mathcal Z}|p(z)-q(z)|.
 \end{aligned}
\end{equation}
Thus $0\le H^2(p,q)\le1$.  For a fixed epoch $e$ and a declared pair
$\theta,\theta'$, define the finite-distance compression loss of a record $T$ by
\begin{equation}\label{en:eq:finite-record-residual-setting}
 R_{T,e}^H(\theta,\theta')=
 H^2(p_{e,\theta},p_{e,\theta'})-
 H^2(T_\#p_{e,\theta},T_\#p_{e,\theta'}).
\end{equation}
Hellinger data processing gives
$0\le R_{T,e}^H(\theta,\theta')\le H^2(p_{e,\theta},p_{e,\theta'})$.
This residual measures distinction lost by recording $T$ after measurement; it
does not measure distinction already lost between the quantum state and the
measured law.  It can be estimated from the unaggregated counts, with
finite-sample decisions based on the simultaneous confidence bounds introduced
below.
\end{definition}

\begin{definition}[Local differential record audit]
\label{en:def:local-differential-setting}
Let $U\subseteq\Theta$ be an open coordinate neighborhood on which
$\theta\mapsto p_{e,\theta}(y)$ is differentiable and
$p_{e,\theta}(y)>0$ for every $y\in\mathcal Y$.  Define
\begin{equation}\label{en:eq:fisher-setting}
 \begin{aligned}
 s_{e,\theta}(y)&:=\nabla_\theta\log p_{e,\theta}(y),\\
 F_Y(e,\theta)&:=\sum_{y\in\mathcal Y}p_{e,\theta}(y)
 s_{e,\theta}(y)s_{e,\theta}(y)^\top,\\
 F_T(e,\theta)&:=F_{T_\#p_e}(\theta).
 \end{aligned}
\end{equation}
The positive-semidefinite difference $F_Y(e,\theta)-F_T(e,\theta)$ is the
infinitesimal sensitivity lost by the record $T$.  This audit is used when the
question concerns local control response rather than a declared finite pair of
settings.  Points with zero probability are excluded from this regular
definition and are treated, when required, by a separately specified
fixed-support limit.
\end{definition}

\begin{definition}[Optional all-direction response-ideal audit]
\label{en:def:all-direction-setting}
Fix $\theta_0\in\Theta$ and a real-analytic coordinate chart about $\theta_0$.
Let $\mathcal O_{\theta_0}$ be the ring of real-analytic function germs at
$\theta_0$.  For an analytic vector response $f=(f_1,\ldots,f_m)$, define its
response ideal by
\begin{equation}\label{en:eq:response-ideal-setting}
 \mathfrak I[f]=\langle f_j(\theta)-f_j(\theta_0):1\le j\le m\rangle
 \subset\mathcal O_{\theta_0}.
\end{equation}
Choose local real coordinates for $\sigma_\theta$.  Let
$W_\theta=w(p_{e,\theta})$, where $w$ is a declared analytic certificate map on
a neighborhood of $p_{e,\theta_0}$.  Define
\begin{equation}\label{en:eq:ideal-chain-setting}
 \begin{aligned}
 \mathfrak I_\rho&=\mathfrak I[\sigma],&
 \mathfrak I_P&=\mathfrak I[p_e],\\
 \mathfrak I_Q&=\mathfrak I[T_\#p_e],&
 \mathfrak I_C&=\mathfrak I[(T_\#p_e,W)],\\
 \mathfrak I_Q&\subseteq\mathfrak I_C
 \subseteq\mathfrak I_P\subseteq\mathfrak I_\rho.
 \end{aligned}
\end{equation}
The first two inclusions follow from the declared analytic postprocessing, and
the last follows from the epoch-fixed channel and measurement.  These ideals
are used only for the optional question of whether the declared response is
complete along every analytic control arc.  The integral-closure criterion and
its assumptions are stated in the completeness theorem below.  This audit is
unnecessary when a finite Hellinger residual already requires refusal.
\end{definition}

If the physical factorization in \cref{en:eq:epoch-born-law} is unavailable,
the observed black-box family
$\mathcal R_e=\{r_{e,\theta}:\theta\in\Theta\}$ still supports a
distribution-level audit.  Pushforwards and finite Hellinger residuals require
only the observed laws.  Fisher matrices additionally require a differentiable
positive response family, and response ideals require real-analyticity in a
specified local chart.  Without a common epoch-fixed channel and measurement,
however, there is no identified state ideal $\mathfrak I_\rho$, and the
inclusion $\mathfrak I_P\subseteq\mathfrak I_\rho$ cannot be asserted.  The
audit may therefore detect response lost by a record, but it cannot attribute
that loss to state preparation, channel evolution, or measurement.  This
separates the model-free distribution audit from the fixed-channel
state--measurement completeness result.

\paragraph{Decision terminology.}
A \emph{record} is a declared deterministic statistic, or more generally a
fixed Markov kernel, applied to the measured outcome.  \emph{Completeness} is
always relative to specified responses. For a finite audit, these are the
declared control pairs together with a Hellinger-residual tolerance; finite-shot approval
uses the confidence procedure introduced below.  For the optional
all-direction audit, the responses are the centered state or distribution
germs, and completeness is equality of the specified integral closures under
the assumptions of the completeness theorem.  A \emph{repair}
is an executable augmentation selected from the declared candidate class that
strictly reduces a detected loss.  It is called a \emph{completion} only after
the target criterion is met.  If neither the original record nor an admissible
augmentation can be certified, the procedure refuses low-dimensional
compression and retains the unaggregated outcome law.

\subsection{End-to-end audit for one three-qubit circuit}
\label{en:sec:running-example}
This example follows one candidate record from the quantum state to an
operational decision.  The full measured law supplies the reference response,
and the Hellinger residual measures the part removed by the record.  Comparing
the quantum, measured, and recorded Fisher matrices then distinguishes
sensitivity already inaccessible under the chosen measurement from sensitivity
lost only by coarse-graining.  The count model identifies the quantities that
must be estimated on hardware.  These distribution-level and local
calculations constitute Steps~1--5 and require neither a fitted microscopic
noise model nor an exponential-family assumption.  The response-ideal
calculation at the end of the example is used only for the optional Steps~6--7,
when completeness along every analytic control arc is at issue.

The objects in Definition~\ref{en:def:epoch-experiment} take the following
form.  The control is $\theta=(a,b)\in\mathbb T^2$, the Hilbert space and the
initial state are
\begin{equation}\label{en:eq:running-environment}
 \mathcal H=(\C^2)^{\otimes3},\qquad
 \rho_0=|100\rangle\!\langle100|,
\end{equation}
and the compiler family is
\begin{equation}\label{en:eq:running-compiler}
 \mathcal C_{(a,b)}(\rho)=U(a,b)\rho U(a,b)^\dagger,\qquad
 U(a,b)=G_{02}(b)G_{01}(a).
\end{equation}
The outcome and computational-bitstring spaces are both
$\mathcal Y=\mathcal X=\{0,1\}^3$.  For the ideal reference, take
$\mathcal E_{\mathrm{id}}=\operatorname{id}$ and
$M_y=|y\rangle\!\langle y|$, and denote the resulting law by $p_{a,b}$.  In
hardware epoch $e$, the same compiled circuit is followed by the unknown
epoch-fixed channel $\mathcal E_e$ and physical POVM $\mathcal M_e$, giving
\begin{equation}\label{en:eq:running-hardware-law}
 p_{e,a,b}(y)=\Tr\!\left[M_{e,y}\mathcal E_e
 \bigl(\mathcal C_{(a,b)}(\rho_0)\bigr)\right].
\end{equation}
This notation keeps the ideal law used for the closed-form calculation distinct
from the observed law to which the audit is applied.  The compression audit is
evaluated directly from the latter probabilities, or from their counts, and
does not require a fitted parameterization of a microscopic noise mechanism.

Let $G_{ij}(\varphi)$ be the excitation-number-preserving two-mode rotation
that acts as the identity on $|00\rangle_{ij}$ and $|11\rangle_{ij}$ and
satisfies
\[
 \begin{aligned}
 G_{ij}(\varphi)|10\rangle_{ij}
 &=\cos\varphi|10\rangle_{ij}+\sin\varphi|01\rangle_{ij},\\
 G_{ij}(\varphi)|01\rangle_{ij}
 &=-\sin\varphi|10\rangle_{ij}+\cos\varphi|01\rangle_{ij}.
 \end{aligned}
\]
For the three qubits $V=\{0,1,2\}$, the one-excitation set is
$S=\{100,010,001\}$.  Applying $U(a,b)=G_{02}(b)G_{01}(a)$ to
$|100\rangle$ gives
\begin{equation}\label{en:eq:running-state}
 |\psi(a,b)\rangle=\cos a\cos b|100\rangle+\sin a|010\rangle
 +\cos a\sin b|001\rangle .
\end{equation}
Figure~\ref{en:fig:running-circuit} fixes the bit order and the order of
operations.  The second rotation acts logically on modes $0$ and $2$ and leaves
qubit $1$ unchanged.  Its box spans three wires only to display this
nonadjacent logical operation; it does not denote a native three-qubit gate.
\begin{figure}[H]
\centering
\begin{quantikz}[row sep=0.27cm,column sep=0.42cm]
\lstick{$q_0$}&\gate{X}&\gate[wires=2]{G_{01}(a)}&\gate[wires=3]{G_{02}(b)\otimes I_1}&\meter{}&\cw\\
\lstick{$q_1$}&\qw& & &\meter{}&\cw\\
\lstick{$q_2$}&\qw&\qw& &\meter{}&\cw
\end{quantikz}
\caption{The three-qubit running circuit.  It prepares one excitation,
redistributes it by $G_{01}(a)$ and $G_{02}(b)$, and records the computational
bitstring.  Its ideal law assigns zero probability outside
$S=\{100,010,001\}$.}
\label{en:fig:running-circuit}
\end{figure}
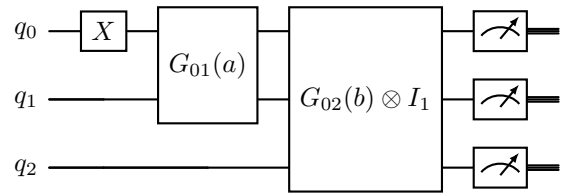
Because $U(a,b)$ is unitary, the state is normalized; explicitly,
$\cos^2a\cos^2b+\sin^2a+\cos^2a\sin^2b=1$.  Computational-basis measurement
assigns the two projectors $|0\rangle\!\langle0|$ and
$|1\rangle\!\langle1|$ to each qubit, so one shot returns a bitstring
$y\in\{0,1\}^3$.  By the Born rule, the ideal probabilities are the squared
amplitudes,
\begin{equation}\label{en:eq:running-probabilities}
 \begin{aligned}
 p_{a,b}(100)&=\cos^2a\cos^2b,\\
 p_{a,b}(010)&=\sin^2a,\\
 p_{a,b}(001)&=\cos^2a\sin^2b.
 \end{aligned}
\end{equation}
For $y\notin S$, $p_{a,b}(y)=0$.  At parameter values where one of these
amplitudes vanishes, the support is a proper subset of $S$.
Thus the three observed counts in an ideal $n$-shot experiment satisfy
\begin{equation}\label{en:eq:running-count-law}
 \begin{aligned}
 &(N_{100},N_{010},N_{001})\\
 &\quad\sim\operatorname{Multinomial}\!\left(
 n;\cos^2a\cos^2b,\sin^2a,\cos^2a\sin^2b\right).
 \end{aligned}
\end{equation}
On hardware, the eight-component count vector obeys
$C_{e,a,b}\sim\operatorname{Multinomial}(n,p_{e,a,b})$ and is retained before
aggregation by any candidate record.  Retaining these unaggregated counts is
what permits the later calculation of variation hidden by a record.  At an
interior point of the positive-amplitude chart, the three ideal score vectors
are
\begin{equation}\label{en:eq:running-scores}
 \begin{aligned}
 s_{a,b}(100)&=(-2\tan a,-2\tan b),\\
 s_{a,b}(010)&=(2\cot a,0),\\
 s_{a,b}(001)&=(-2\tan a,2\cot b),
 \end{aligned}
\end{equation}
and direct substitution into \cref{en:eq:fisher-setting} yields
\begin{equation}\label{en:eq:running-fisher}
 F_{\rm cb}(a,b)=
 \begin{pmatrix}4&0\\0&4\cos^2a\end{pmatrix}=F_Q(a,b).
\end{equation}
The equality with the pure-state QFI follows because the amplitudes in
\cref{en:eq:running-state} are real, positive, and have no parameter-dependent
phase in this chart.  Hence, for the ideal reference at each such interior
point, measurement loss and record-compression loss separate as
\begin{equation}\label{en:eq:running-loss-separation}
 \begin{aligned}
 \underbrace{F_Q-F_{\rm cb}}_{\text{measurement loss}}&=0,\\
 \underbrace{F_{\rm cb}-F_{T_1}}_{\text{compression loss}}
 &=\begin{pmatrix}0&0\\0&4\cos^2a\end{pmatrix}.
 \end{aligned}
\end{equation}
Thus a missing $b$ response under $T_1$ cannot be attributed to the ideal
computational-basis measurement in this chart; it is introduced by merging
measured outcomes.  This zero measurement-loss statement concerns the ideal
reference and is not assumed for the hardware law $p_{e,a,b}$.

\paragraph{Which outcomes the candidate records actually merge.}
On $S=\{100,010,001\}$, the two fibers of the one-site record $T_1=n_1$ are
\begin{equation}\label{en:eq:running-t1-fibers}
 T_1^{-1}(0)\cap S=\{100,001\},\qquad
 T_1^{-1}(1)\cap S=\{010\}.
\end{equation}
Thus $T_1$ merges $100$ and $001$.  By contrast, the augmented record
$T_2=(L,n_1,n_2)$ gives
\begin{equation}\label{en:eq:running-t2-values}
 100\mapsto(0,0,0),\qquad
 010\mapsto(0,1,0),\qquad
 001\mapsto(0,0,1),
\end{equation}
and is injective on $S$.  Low dimension alone does not establish preservation:
the relevant question is whether outcomes placed in the same fiber have the
same declared control response.

With $n_j(y)=y_j$, the ideal pushforward under the one-site record $T_1=n_1$
and its Fisher matrix are
\begin{equation}\label{en:eq:t1-pushforward}
 T_{1\#}p_{a,b}=\operatorname{Bernoulli}(\sin^2a),\qquad
 F_{T_1}(a,b)=\begin{pmatrix}4&0\\0&0\end{pmatrix}.
\end{equation}
It therefore erases the entire $b$ direction.  This is also visible at finite
separation.  For $a$ fixed and $b,b'$ in the positive chart,
\begin{equation}\label{en:eq:hidden-b-hellinger}
 \begin{aligned}
 H^2(p_{a,b},p_{a,b'})
   &=\cos^2a\,[1-\cos(b-b')],\\
 H^2(T_{1\#}p_{a,b},T_{1\#}p_{a,b'})&=0.
 \end{aligned}
\end{equation}
Hence the Hellinger residual hidden by $T_1$ equals the full first expression.
Its local expansion is
$\frac12\cos^2a(b-b')^2+O(|b-b'|^4)$, consistent with the Fisher--Hellinger
relation $H^2\sim\frac18F_{bb}(b-b')^2$.

\paragraph{Auditing hardware response without a fitted noise model.}
For arbitrary hardware settings $\theta=(a,b)$ and $\theta'=(a',b')$, define
\begin{equation}\label{en:eq:running-operational-residual}
 R_{T,e}^H(\theta,\theta')=
 H^2(p_{e,\theta},p_{e,\theta'})-
 H^2(T_\#p_{e,\theta},T_\#p_{e,\theta'})\ge0.
\end{equation}
The first term is the full change displayed by the device, and the second is
the change retained by the proposed record.  Hence $R_{T,e}^H$ measures
observed variation hidden by recording, not the microscopic cause of that
variation.  From the unaggregated counts, compute
\begin{equation}\label{en:eq:running-count-residual}
 \widehat p_{e,\theta}=N_{e,\theta}/n,\qquad
 N^T_{e,\theta}(t)=\sum_{y:T(y)=t}N_{e,\theta}(y),
\end{equation}
and take the difference of the two empirical Hellinger distances.  Because the
same empirical histograms are aggregated by the deterministic map $T$, data
processing makes the raw plug-in difference nonnegative for every sample.
This algebraic nonnegativity is not a confidence statement: finite-shot
uncertainty enters the simultaneous bounds used for approval or refusal.  The
calculation neither assumes that $p_{e,\theta}$ is an exponential family nor
identifies Kraus parameters of $\mathcal E_e$.

On the allowed set, the two-occupation record $T_S=(n_1,n_2)$ is injective
because $n_0=1-n_1-n_2$.  The experimentally used augmentation is
$T_2=(L,n_1,n_2)$.  Thus $T_2$ retains the complete histogram on $S$ together
with a separate flag for outcomes outside $S$.  If an
ideal law $p$ is supported on $S$ and a hardware law decomposes as
$r=u q+(1-u)h$ with $q$ supported on $S$ and $h$ on $S^c$, then
\begin{equation}\label{en:eq:running-hellinger-decomposition}
 H^2(p,r)=(1-\sqrt u)+\sqrt u\,H^2(p,q).
\end{equation}
The first term is the exact contribution normal to the allowed face of the
probability simplex.  Coarse-graining $q$ by a record $T$ leaves the
nonnegative within-support residual
\begin{equation}\label{en:eq:running-residual}
 R_T^H=H^2(p,q)-H^2(T_\#p,T_\#q)\ge0.
\end{equation}
For $T_1$ and the comparison in \cref{en:eq:hidden-b-hellinger}, this residual
equals the full distance; for $T_S$ it is zero.

Injectivity of $T_2$ on $S$ does not imply sufficiency for every noisy hardware
response.  The flag $L$ records only the total mass of $S^c$ and does not
distinguish individual off-support bitstrings.  If their conditional
distribution changes with control, $T_2$ retains a residual.  The audit does
not assume this distribution to be parameter independent; it tests the
residual through \cref{en:eq:running-operational-residual}.  Injectivity is an
exact property of the ideal allowed law, whereas a small hardware residual can
support approval only for the declared, observed comparisons.

\paragraph{Optional all-direction calculation.}
Steps~1--5 end with the finite residual decision above.  To ask whether the
same conclusion holds along every analytic arc through $(a,b)=(0,0)$, take
Taylor expansions and form
\begin{align}\label{en:eq:running-ideals}
 \mathfrak I_\rho&=(a,b),& \mathfrak I_P&=(a^2,b^2),\nonumber\\
 \mathfrak I_{Q_1}&=(a^2),& \mathfrak I_{Q_2}&=(a^2,b^2).
\end{align}
The first equality uses a projective amplitude chart; the remaining three use
the computational histogram.  Thus $T_1$ fails on the $b$ axis, whereas $T_2$
is complete relative to that measurement.  It is not complete for the state
response, because the interference observables
the interference observables
\begin{equation}\label{en:eq:running-interference}
 \begin{aligned}
 X_{01}&=|100\rangle\!\langle010|+|010\rangle\!\langle100|,\\
 X_{02}&=|100\rangle\!\langle001|+|001\rangle\!\langle100|.
 \end{aligned}
\end{equation}
have expectations
$\langle X_{01}\rangle=\sin(2a)\cos b=2a+O(\|(a,b)\|^2)$ and
$\langle X_{02}\rangle=\cos^2a\sin(2b)=2b+O(\|(a,b)\|^2)$.  Adding these
coordinates changes the certificate ideal to $(a,b)$ and exposes the signed
first-order amplitude response hidden by computational probabilities at the
boundary.

The target must remain explicit when the completion theorem is applied.  In
the local analytic ring at $(0,0)$, take the one-site record as the current
representation and the computational histogram as the first target:
\begin{equation}\label{en:eq:running-library-completion}
 \mathfrak I_0=(a^2),\qquad
 \mathfrak I_\star=\mathfrak I_P=(a^2,b^2),\qquad
 \mathfrak J_{n_2}=(b^2).
\end{equation}
Then $\mathfrak I_0+\mathfrak J_{n_2}=(a^2,b^2)=\mathfrak I_\star$, so adding
$n_2$ completes the record relative to computational readout. For the stronger
state-response target,
\begin{equation}\label{en:eq:running-state-completion}
 \mathfrak I_\star=\mathfrak I_\rho=(a,b),\qquad
 \mathfrak J_{X_{01}}=(a),\qquad
 \mathfrak J_{X_{02}}=(b),
\end{equation}
and $\mathfrak I_P+\mathfrak J_{X_{01}}+\mathfrak J_{X_{02}}=(a,b)$.
The first equality certifies record augmentation under the same measurement.
The second requires new interference settings and is operational only when
those settings belong to the admissible hardware library.  Completion relative
to one measurement therefore does not imply completion for the state response.

\paragraph{A phase extension isolating measurement loss.}
To distinguish measurement loss from record compression, attach a relative
phase to one occupied mode:
\begin{equation}\label{en:eq:running-phase-extension}
 \begin{aligned}
 |\psi(a,b,\phi)\rangle={}&\cos a\cos b|100\rangle
   +e^{i\phi}\sin a|010\rangle\\
 &+\cos a\sin b|001\rangle.
 \end{aligned}
\end{equation}
Its computational probabilities are independent of $\phi$, whereas at an
interior point
\begin{equation}\label{en:eq:running-phase-loss}
 F_{{\rm cb},\phi\phi}=0,
 \qquad F_{Q,\phi\phi}=4\sin^2a\cos^2a=\sin^2(2a).
\end{equation}
No statistic of the same computational-basis counts can recover this missing
direction, because it is absent before compression.  With
$Y_{01}=-i|100\rangle\!\langle010|+i|010\rangle\!\langle100|$, however,
\begin{equation}\label{en:eq:running-phase-witness}
 \begin{aligned}
 A&=\sin(2a)\cos b,\\
 \langle X_{01}\rangle&=A\cos\phi,\\
 \langle Y_{01}\rangle&=A\sin\phi.
 \end{aligned}
\end{equation}
Thus an interference setting supplies a local $\phi$ witness whenever that
phase is physically identifiable.  The two losses are now operationally
distinct:
$T_1$ loses $b$ after measurement, while computational readout loses $\phi$
before any record is formed.

\paragraph{Interpreting the phase loss through attainable geometry.}
On the slice $b=0$, an ideal $Y_{01}$-basis measurement attains
$g_{M,Y}(\partial_\phi,\partial_\phi)=\sin^2(2a)=
g_Q(\partial_\phi,\partial_\phi)$ at $\phi=0$.  If interference visibility is
$0\le\eta\le1$ and the two probabilities become
$[1\pm\eta\sin(2a)\sin\phi]/2$, the attainable value is
$\eta^2\sin^2(2a)$.  If the computational basis is the only admissible setting,
then $h_{\rm att}(\partial_\phi)=0$: phase sensitivity is inaccessible under
that device budget, so record augmentation is not an operational repair.  If a
calibrated $Y_{01}$ setting is admissible, then
$h_{\rm att}(\partial_\phi)\ge\eta^2\sin^2(2a)>0$ and measurement augmentation
is actionable.  A nonzero $F_Q-F_M$ alone therefore does not justify a
follow-up measurement; the current $F_M$ must be compared with the attainable
metric fibre under the declared device budget.

\paragraph{Operational decision for the example.}
The four parts of the calculation are used in the following order.
\begin{enumerate}[label=(\roman*),leftmargin=2.1em]
\item \emph{Distribution-level hardware audit.}  Retain the full count vector, form
$\widehat p_{e,\theta}$ and $T_\#\widehat p_{e,\theta}$ from the same shots, and
bound \cref{en:eq:running-operational-residual}.  Neither a Kraus noise model
nor an exponential-family law is fitted.
\item \emph{Compression loss and its missed direction.}  For $T_1$, both
\cref{en:eq:hidden-b-hellinger,en:eq:running-loss-separation} identify the
missing direction as $v=(0,1)$: the full law responds to $b$, while the quotient
does not.  Hence $T_1$ is refused.
\item \emph{Measurement loss.}  The equality $F_Q=F_{\rm cb}$ shows that no
interior $(a,b)$ information was lost by computational readout.  By contrast,
\cref{en:eq:running-phase-loss} shows a genuine $\phi$ direction already absent
from those counts; it requires the interference witnesses in
\cref{en:eq:running-phase-witness}, not a richer statistic of the same data.
Such a witness is called an executable augmentation only after it is shown to
belong to the admissible setting set and to increase $h_{\rm att}$.
\item \emph{Repair, conditional approval, or refusal.}  The augmentation
$T_2$ repairs the allowed histogram, but it is accepted for the observed
hardware comparisons only if a simultaneous upper confidence bound for
$R_{T_2}^H$ lies inside the declared tolerance.  If its off-support residual or
a missed witness direction remains, the output is the full histogram together
with a refusal, rather than the least inadequate low-dimensional record.
\end{enumerate}
The example compares full and compressed responses computed from the same data
without declaring a sufficient statistic in advance.  The finite Hellinger
residual, local Fisher loss, and response ideals detect the same $b$-direction
failure at different resolutions.  This agreement is established for the
present circuit and does not imply universal sufficiency for other circuit
families.

\section{Theoretical development: response completeness and auditable refusal}
\label{en:sec:theory}
Section~2 separated measurement loss from record-compression loss, identified
a missed direction, and returned refusal when neither the record nor an
admissible augmentation met its target.  We now establish these operations for
a general finite response family.  No exponential-family model is imposed on
the hardware law.  A structured exponential form is introduced later only as
a verifiable sufficient regime for exact preservation.

The theory addresses one question: when does a record preserve the declared
responses, and when must compression be refused? Response ideals
formulate preservation along analytic control arcs; a finite Rees test provides
an algebraic witness when that preservation fails; an admissible library may
repair the witnessed deficit or establish relative failure within that
library; and simultaneous finite-shot bounds implement the corresponding
decision for observed comparisons.  Conditional-score identities, data
processing, the QFI covariance identity, and Rees theory are used as established
tools. Here they are applied to the same completion-or-refusal decision.

The conclusions have three scopes. A finite observed
comparison supports approval, refusal, or deferral only at its declared
tolerance.  A differentiable local model additionally supports Fisher or score
statements in specified tangent directions.  Equality of the relevant
response-ideal closures establishes preservation of first response order along
every admissible analytic arc. The hardware experiments
in Section~4 test the first level and selected directions of the second, not
the all-arc conclusion of the third.

The calculation proceeds in the following order. The results below state the
conditions required at each step and the conclusion that follows.

\paragraph{Computation protocol: from counts to a supported decision.}
\begin{enumerate}[label=\textbf{\arabic*.},leftmargin=2.5em]
\item \emph{Freeze the comparison.}  Specify the epoch, controls, finite outcome
set, candidate record or fixed Markov kernel, comparison family, tolerances,
and familywise error level before inspecting the selected contrasts.
\item \emph{Retain the full response.}  Store the full count vector at every
control and form each compressed histogram from those same counts.
\item \emph{Audit finite and local loss.}  Compute the full-minus-compressed
Hellinger residual.  When a differentiable model or local fit is justified,
also compute $F_Y-F_T$; a positive eigenvector gives a missed local direction.
\item \emph{Separate the stage and actionability of loss.}  Compare $F_Q-F_Y$
only when the state family and fixed measurement are available.  A nonzero
$F_Y-F_T$ calls for a richer record.  A nonzero $F_Q-F_Y$ calls for a new
measurement only if an admissible setting increases the directional support
$h_{\rm att}$; otherwise report the direction as inaccessible under the current
device budget.
\item \emph{Make separate finite-sample decisions.}  Optimize total distance
$D$ and hidden residual $R$ over one simultaneous confidence region.  Use $R$
to approve, reject, or defer the record and $D$ to pass, fail, or defer the task.
Form a combined positive claim only when both upper bounds meet their distinct
budgets; a task failure alone is not a record failure.
\item \emph{Audit an all-direction claim only when needed.}  Form the centered
ideals $\mathfrak I_\rho\supseteq\mathfrak I_P\supseteq\mathfrak I_C$ and
compare the finitely many Rees valuations of $\mathfrak I_C$.  Equality
certifies first-response order relative to the declared analytic germ; a
strict inequality is algebraic evidence for refusal.
\item \emph{Lift, test, and report the witness.}  On the normalized blow-up,
test real accessibility and admissibility, descend a transverse arc, and run
that control path.  If this lifting cannot be certified, report
``algebraic-only'' rather than claiming an executable direction.
\end{enumerate}
Steps 1--5 are the default hardware procedure. Steps 6--7 are required only
for the stronger all-direction completeness claim; they are not prerequisites
for rejecting a record on an observed comparison.

\subsection{Measurement loss and compression loss}
Let $\theta\mapsto\rho_\theta$ be a real-analytic state family, let a fixed POVM
produce $p_\theta\in\Delta(\mathcal Y)$, and let
$T:\mathcal Y\to\mathcal Z$ be a deterministic record.  At a positive regular
point write $S_Y=\nabla_\theta\log p_\theta(Y)$.  The score of the pushforward
is $S_T(T)=\E[S_Y\mid T]$, and therefore
\begin{proposition}[Exact local compression loss]\label{en:prop:score}
\begin{equation}\label{en:eq:compact-score-loss}
 F_Y-F_T=\E\!\left[\Cov\{S_Y(Y)\mid T(Y)\}\right]\succeq0.
\end{equation}
Equality holds in a direction $v$ exactly when $v^\top S_Y$ is constant on
every $T$-fiber of positive probability.  Thus an eigenvector of
$F_Y-F_T$ with positive eigenvalue is a locally missed control direction.
\end{proposition}
This is the standard conditional-score identity.  Here it separates loss
created by recording after readout from loss already created by the POVM
\cite{AyJostLeSchwachhofer2015}.  For a pure-state amplitude--phase model,
\begin{equation}\label{en:eq:compact-phase-gap}
 \begin{aligned}
 |\psi_\theta\rangle
   &=\sum_y\sqrt{p_\theta(y)}e^{i\phi_\theta(y)}|y\rangle,\\
 F_Q-F_Y&=4\Cov_{p_\theta}(d\phi,d\phi)\succeq0.
 \end{aligned}
\end{equation}
Hence $F_Q-F_Y$ and $F_Y-F_T$ diagnose different stages.  Refining a classical
record can reduce the second gap but cannot in general repair the first.

The deterministic record is not essential to this local layer.  Let
$K(z\mid y)$ be a fixed Markov kernel, let $Z\mid Y=y\sim K(\cdot\mid y)$, and
write $Kp_\theta$ for the observed law of $Z$.
\begin{proposition}[Fixed randomized postprocessing]\label{en:prop:markov-kernel}
At every positive regular point,
\begin{equation}\label{en:eq:markov-score}
 \begin{aligned}
 S_Z(Z)&=\E[S_Y(Y)\mid Z],\\
 F_Y-F_{Kp}&=\E\!\left[\Cov\{S_Y(Y)\mid Z\}\right]\succeq0.
 \end{aligned}
\end{equation}
Moreover $H^2(Kp,Kq)\le H^2(p,q)$, and the response-ideal inclusion used below
remains valid because $p\mapsto Kp$ is linear.  Thus the distribution-level
local, finite-distance, and response-ideal audits extend from deterministic
statistics to fixed randomized readout, binning, or classical communication.  If $K$
depends on the control, its derivative contributes an additional response and
must be included in the full measured model; it cannot be treated as harmless
postprocessing.
\end{proposition}
\begin{proof}
The joint score of $(Y,Z)$ is $S_Y(Y)$ because $K$ is parameter independent.
Conditioning it on $Z$ gives the first identity, and total covariance gives the
second.  Hellinger contraction is data processing for Markov kernels.  The last
claim follows coordinatewise from $(Kp)_z=\sum_yK(z\mid y)p_y$.
\end{proof}

Expanding the pure-state QFI gives the amplitude term $F_Y$ and the displayed
phase covariance; the real cross terms cancel.  In a specified tangent
direction, the phase term vanishes exactly when the infinitesimal phase is
constant on the support and is therefore only global.  This identifies the
information boundary of one fixed computational-basis measurement.

For all-direction comparison, fix a control germ $(\Theta,\theta_0)$ with
local analytic ring $R$.  Let $\mathfrak I_\rho$, $\mathfrak I_P$, and
$\mathfrak I_C$ be generated by the centered coordinates of the state, the
full measured histogram, and an augmented certificate $C=(T_\#p,W)$,
respectively.  Because the POVM is fixed and $C$ is analytic in the histogram,
\begin{equation}\label{en:eq:compact-ideal-chain}
 \mathfrak I_C\subseteq\mathfrak I_P\subseteq\mathfrak I_\rho.
\end{equation}
For an analytic arc $\gamma$, let
$\nu_I(\gamma)=\operatorname{ord}_t\gamma^*I$.

The next equivalence is an application, not an extension, of integral-closure
theory.  For one inclusion it is the classical analytic-arc criterion
\cite{LejeuneTeissier2008,HunekeSwanson2006}.  Its contribution here is the
choice of three response ideals induced by a fixed quantum measurement and a
record, and the resulting two adjacent obstructions with different operational
repairs.

\begin{theorem}[Two-level response completeness]
\label{en:thm:two-level-completeness}\label{en:thm:integral-completeness}
Assume that the chosen real-analytic germ and its admitted arc class satisfy
the analytic-arc criterion for integral dependence.
\begin{enumerate}[label=(\alph*),leftmargin=2.2em]
\item $\overline{\mathfrak I_C}=\overline{\mathfrak I_P}$ if and only if
$\nu_C(\gamma)=\nu_P(\gamma)$ for every nonconstant analytic control arc.
This is completeness relative to the recorded measurement response.
\item $\overline{\mathfrak I_P}=\overline{\mathfrak I_\rho}$ if and only if
the fixed measurement preserves the first nonzero state-response order along
every such arc.
\item $\overline{\mathfrak I_C}=\overline{\mathfrak I_\rho}$ holds if and
only if both adjacent losses are absent.  Classical postprocessing cannot
repair a direction already annihilated by the measurement.
\end{enumerate}
\end{theorem}
\begin{proof}
Ideal inclusion in \cref{en:eq:compact-ideal-chain} reverses arc order.  The
analytic arc criterion identifies mutual integral dependence with equality of
those orders on every arc; apply it to the two adjacent inclusions and use
transitivity of integral closure~\cite{LejeuneTeissier2008,HunekeSwanson2006}.
\end{proof}

The state ideal is independent of the chosen operator coordinates because an
operator-basis change is an
invertible real-linear transformation of its generators.  The histogram ideal
uses any nonredundant probability coordinates, and an analytic local bijection
of certificate coordinates leaves $\mathfrak I_C$ unchanged.  The theorem
therefore compares the generated response algebras rather than arbitrary
coordinate lists.

For a finite comparison, the corresponding observable defect is
\begin{equation}\label{en:eq:compact-hellinger-residual}
 R_T^H(\theta,\theta')=
 H^2(p_\theta,p_{\theta'})-
 H^2(T_\#p_\theta,T_\#p_{\theta'})\ge0.
\end{equation}
The inequality is data processing and also holds sample by sample when the
same empirical histograms are aggregated by $T$.  Sampling uncertainty enters
the confidence bounds for the unknown population residual; it is distinct from
the algebraic nonnegativity of this same-sample plug-in difference.

At a full-support baseline, if
$p_{\gamma(t)}=p_0+a t^\nu+O(t^{\nu+1})$ along an analytic arc, then
\begin{equation}\label{en:eq:compact-shot-order}
 H^2(p_{\gamma(t)},p_0)
 =\frac18\sum_y\frac{a_y^2}{p_0(y)}t^{2\nu}+o(t^{2\nu}).
\end{equation}
Simple-hypothesis discrimination consequently has local shot order
$\Theta(t^{-2\nu}\log(1/\delta))$.  Theorem~\ref{en:thm:two-level-completeness} preserves this exponent
when the certificate readout has a locally comparable Hellinger metric.  Equal
constants and rank-changing support boundaries require separate control.

\subsection{A verifiable structured regime and its boundary}
The general audit above is distribution based and imposes no parametric law.
We now identify a positive regime in
which the proposed record can be justified before hardware testing.  Let
$S$ be finite, $A:S\to\mathbb R^d$, and
\begin{equation}\label{en:eq:compact-positive-family}
 p_\vartheta(x)=h(x)\exp\{\langle\vartheta,A(x)\rangle-\psi(\vartheta)\},
 \qquad x\in S,
\end{equation}
with $h(x)>0$.
\begin{proposition}[Structured positive regime]\label{en:prop:cov}
After quotienting parameter directions on which $A$ is constant,
\begin{equation}\label{en:eq:compact-covariance}
 \nabla\log p_\vartheta(x)=A(x)-\E_\vartheta A,
 \qquad F_Y=\Cov_\vartheta(A).
\end{equation}
For this family, the record $A$ is minimal sufficient up to bijective
relabeling of equal $A$-values.  For the positive-real state
\(\lvert\psi_\vartheta\rangle=\sum_x\sqrt{p_\vartheta(x)}\lvert x\rangle\),
\begin{equation}\label{en:eq:compact-qfi-natural}
 F_Q(\vartheta)=F_Y(\vartheta)=\Cov_\vartheta(A).
\end{equation}
If instead the amplitudes are parameterized as
$\sqrt{h(x)}\exp\{\langle\theta,A(x)\rangle-\psi(2\theta)/2\}$, then
$\vartheta=2\theta$ and the matrix in the $\theta$ coordinates is
$4\Cov_{2\theta}(A)$.  Parameter-dependent relative phases add the
nonnegative gap in \cref{en:eq:compact-phase-gap}.
\end{proposition}

Indeed,
$\partial_i\lvert\psi_\vartheta\rangle
=\frac12\sum_x\sqrt{p_\vartheta(x)}
(A_i(x)-\E_\vartheta A_i)\lvert x\rangle$ and
$\langle\psi_\vartheta\mid\partial_i\psi_\vartheta\rangle=0$; the pure-state
QFI formula gives \cref{en:eq:compact-qfi-natural}.  This explicit convention
prevents the factor four associated with an amplitude parameter from being
attributed to the natural probability parameter.

The scope of this proposition is directly checkable.  For positive laws, $A$ is sufficient on a
connected open parameter set exactly when each likelihood ratio within an
$A$-fiber is parameter independent.  In a monomial amplitude chart this is
equivalent to equality of the relevant exponent differences.  A circuit being
number preserving, blockade respecting, or parameterized is not by itself a
certificate; the fiberwise ratio or score condition must still be verified.

Let $\mathcal Y=S\sqcup S^c$ and $L=\mathbf1\{Y\notin S\}$.  Define
$T_I(y)=(0,A(y))$ for $y\in S$ and $T_I(y)=(1,\ast)$ for $y\in S^c$, where
$\ast$ is one undifferentiated symbol.  Write
\begin{equation}\label{en:eq:compact-support-model}
 r_{\eta}(y)=
 \begin{cases}
 (1-\ell_\eta)p_{\vartheta_\eta}(y),&y\in S,\\
 \ell_\eta s_\eta(y),&y\in S^c.
 \end{cases}
\end{equation}
\begin{proposition}[Support--moment quotient and exact validity]
\label{en:thm:quotient}
If $s_\eta$ is independent of $\eta$, then $T_I$ is sufficient for
\cref{en:eq:compact-support-model}.  With
$u=2\arcsin\sqrt\ell$, $\mu=\E_\vartheta A$, and
$C_\vartheta=\Cov_\vartheta(A)$, its Fisher metric is
\begin{equation}\label{en:eq:compact-quotient-metric}
 g_I=du^2+\cos^2(u/2)d\mu^\top C_\vartheta^+d\mu.
\end{equation}
For an arbitrary smooth hardware response, exact local validity instead holds
if and only if the full score is constant on every $T_I$-fiber.
\end{proposition}
\begin{proof}
Under fixed $s$, the law factorizes through $(L,A)$, proving sufficiency.
The support Bernoulli score is orthogonal to the conditional in-support score;
$d\mu=C_\vartheta d\vartheta$ then gives the displayed metric.  The last
statement is Proposition~\ref{en:prop:score} applied to the fibers of $T_I$.
\end{proof}

For approximate validity define, on the identifiable tangent,
\begin{equation}\label{en:eq:compact-relative-defect}
 \Delta_I=F_Y-F_{T_I},\qquad
 \rho_I=\lambda_{\max}
 \bigl(F_Y^{+/2}\Delta_I F_Y^{+/2}\bigr).
\end{equation}
If $\rho_I\le\varepsilon<1$, then
$(1-\varepsilon)F_Y\preceq F_{T_I}\preceq F_Y$.  Every directional Fisher
information is retained by at least $1-\varepsilon$, and an associated scalar
Cram\'er--Rao lower bound grows by at most $(1-\varepsilon)^{-1}$.  This uses
the observed score defect rather than extrapolating fixed off-support shape
from the ideal model.

For a target $q$ supported on $S$ and
$r=(1-\ell)p+\ell s$, disjoint support gives the exact finite decomposition
\begin{proposition}[Support--moment Hellinger diagnostic]
\label{en:thm:diagnostic}
\begin{equation}\label{en:eq:compact-support-hellinger}
 H^2(q,r)=1-\sqrt{1-\ell}
 +\sqrt{1-\ell}\,H^2(q,p).
\end{equation}
The first term is a support-normal floor.  Any further aggregation of $p$ by
$A$ leaves the nonnegative within-support residual
$H^2(q,p)-H^2(A_\#q,A_\#p)$.
\end{proposition}
\begin{proof}
The Hellinger affinity splits over $S$ and $S^c$.  Since $q(S^c)=0$, only
$\sqrt{1-\ell}\sum_{y\in S}\sqrt{q(y)p(y)}$ remains.  Subtracting from one
gives the identity, and data processing gives the nonnegative residual.
\end{proof}

\begin{proposition}[Why exact structured validity is not generic]
\label{en:prop:no-go}
If a record $T$ has a fiber containing two positive-probability outcomes, then
the unrestricted simplex contains a tangent that transfers mass inside that
fiber.  Its pushforward tangent is zero while its full Fisher norm is positive.
Consequently no noninjective low-dimensional record is universally sufficient
for arbitrary hardware responses.  Control-dependent off-support shape or
within-$A$-fiber motion must be tested through the residual rather than assumed
away.
\end{proposition}
\begin{proof}
Choose $y_1,y_2$ in one fiber and a tangent with components $c,-c$ on this
pair and zero elsewhere.  Its pushforward is zero, whereas its Fisher norm is
$c^2/p(y_1)+c^2/p(y_2)>0$.
\end{proof}

\subsection{A finite algebraic witness and refusal rule}
The all-arc statement in Theorem~\ref{en:thm:two-level-completeness} has a finite algebraic audit.  Assume
$R$ is an analytically unramified Noetherian local domain and
$0\ne\mathfrak I_C\subseteq\mathfrak I_P$ are proper.  Let
$v_1,\ldots,v_s$ be the Rees valuations of $\mathfrak I_C$.
\begin{theorem}[Finite algebraic refusal witness]
\label{en:thm:finite-rees-audit}
\begin{equation}\label{en:eq:compact-rees-audit}
 \overline{\mathfrak I_C}=\overline{\mathfrak I_P}
 \quad\Longleftrightarrow\quad
 v_j(\mathfrak I_C)=v_j(\mathfrak I_P),\quad j=1,\ldots,s.
\end{equation}
If equality fails, a valuation with
$v_j(\mathfrak I_C)>v_j(\mathfrak I_P)$ is a finite certificate of a delayed
compressed response.  This certificate is algebraic.  It yields a physically
executable missed direction only if the corresponding exceptional divisor has
an accessible real point and a transverse real analytic arc descends into the
physical control germ.
\end{theorem}
\begin{proof}
This is the Rees valuation criterion for integral closure applied in both
directions.  A strict valuation inequality is the contrapositive witness; a
transverse arc over the complexification realizes its order after the
corresponding normalized blow-up~\cite{HunekeSwanson2006,DeJong1998}.  The
extra real-accessibility conditions are precisely what is needed to interpret
that algebraic arc as an executable control path.
\end{proof}

\begin{remark}[Algebraic refusal versus executable direction]
The valuation inequality alone is sufficient to refuse an all-direction
completeness claim.  It does not assert that a laboratory control path realizes
the divisor.  An executable witness additionally requires (i) a real point of
the divisor over the baseline, (ii) a real analytic transverse arc through that
point, and (iii) a descended arc contained in the admissible control set.  If
any condition is unavailable, the output is an algebraic refusal certificate,
not a claimed physical direction; the valuation may still guide a subsequent
measurement or control design.
\end{remark}

The distinction can be made constructive when the control germ and admissible
controls are given algebraically or semialgebraically.
\begin{proposition}[Sound witness-lifting procedure]
\label{en:prop:witness-lifting}
Assume that the response generators are polynomial or convergent-algebraic near
$\theta_0$, that a normalized blow-up of $\mathfrak I_C$ is available, and that
the admissible real control set $\Theta_{\rm adm}$ is semialgebraic.  For every
strict Rees witness $v_j(\mathfrak I_C)>v_j(\mathfrak I_P)$, perform:
\begin{enumerate}[label=(\roman*),leftmargin=2.2em]
\item identify the corresponding exceptional divisor $E_j$;
\item test for a smooth real point $x\in E_j$ above $\theta_0$ whose blow-down
has a germ in $\Theta_{\rm adm}$;
\item in a real chart with $E_j=\{z_1=0\}$, choose a transverse arc
$z_1=t$, $z_{r>1}=z_r(x)$, descend it, and verify admissibility and the two
orders directly.
\end{enumerate}
Whenever this procedure returns an arc $\gamma_j$, it is physically executable
within the declared control set and satisfies
$\nu_C(\gamma_j)>\nu_P(\gamma_j)$.  If no accessible real point is certified,
the correct output is \emph{algebraic-only} or \emph{unresolved}, not absence of
a physical witness.  Conversely, under the smooth real-accessibility and
transversality assumptions of Theorem~\ref{en:thm:finite-rees-audit}, the procedure finds a local
executable arc.  Real quantifier elimination makes the feasibility steps
decidable for polynomial data, although the worst-case cost may be large.
\end{proposition}
\begin{proof}
A transverse coordinate has order one on $E_j$, so pullback along the chart arc
realizes the divisorial valuation.  Blow-down preserves the orders of the
response ideals, and the explicit admissibility test places the descended germ
in $\Theta_{\rm adm}$.  The converse follows by choosing local real coordinates
at the assumed smooth accessible point.  Semialgebraic feasibility and order
comparison reduce to finite polynomial sign and ideal-membership tests.
\end{proof}

\begin{theorem}[Library-relative completion or certified failure]
\label{en:thm:accessible-completion}
Let $\mathfrak I_0\subseteq\mathfrak I_\star$ be the current and target
response ideals in an analytically unramified Noetherian local domain, and let
$\mathcal L=\{1,\ldots,L\}$ be a predeclared finite library of executable
augmentations.  Augmentation $\ell$ contributes an ideal
$\mathfrak J_\ell\subseteq\mathfrak I_\star$.  Starting from
$\mathfrak I^{(0)}=\mathfrak I_0$, repeat the following operation:
\begin{enumerate}[label=(\alph*),leftmargin=2.2em]
\item if $\overline{\mathfrak I^{(k)}}=\overline{\mathfrak I_\star}$, stop and
return the selected augmentations;
\item otherwise choose a strict Rees witness $v$ with
$v(\mathfrak I^{(k)})>v(\mathfrak I_\star)$; if an unused $\ell$ satisfies
$v(\mathfrak J_\ell)<v(\mathfrak I^{(k)})$, set
$\mathfrak I^{(k+1)}=\mathfrak I^{(k)}+\mathfrak J_\ell$;
\item if no such $\ell$ exists, stop and return $v$ as a certificate of
impossibility relative to $\mathcal L$.
\end{enumerate}
The procedure is sound and stops after at most $L$ additions.  Moreover,
\begin{equation}\label{en:eq:library-completion}
 \overline{\mathfrak I_0+\sum_{\ell\in\mathcal L}\mathfrak J_\ell}
 =\overline{\mathfrak I_\star}
\end{equation}
if and only if every possible refusal in step (c) is excluded; under
\eqref{en:eq:library-completion} every admissible sequence of strict-witness
improvements terminates with a complete certificate.

The local attainable-metric counterpart has the same relative scope.  Let $U$ be a target tangent
space modulo the null space of $G_Q$, let $G_0\succeq0$ be the current
attainable Fisher metric, and let $G_\ell\succeq0$ be the metric contributed by
library element $\ell$ per declared positive shot allocation.  Put
$K_0=U\cap\ker G_0$.  A Fisher-complete augmentation exists exactly when
\begin{equation}\label{en:eq:fisher-library-kernel}
 K_0\cap\bigcap_{\ell\in\mathcal L}\ker G_\ell=\{0\}.
\end{equation}
If the intersection contains $u\ne0$, then $u$ is a library-relative failure
certificate: no mixture of the declared settings detects that
QFI-positive direction.  Otherwise settings can be chosen successively so
that the current kernel dimension decreases strictly, and at most
$\dim K_0$ settings make the summed metric positive definite on $U$.
\end{theorem}
\begin{proof}
For ideals, $v(I+J)=\min\{v(I),v(J)\}$.  Hence every accepted augmentation
strictly lowers the selected witness value and an already selected library
element can never be needed again.  If
\eqref{en:eq:library-completion} holds while the current ideal is incomplete,
then for any strict witness $v$,
\[
 v(\mathfrak I_\star)
 =v\!\left(\mathfrak I_0+\sum_\ell\mathfrak J_\ell\right)
 =\min\{v(\mathfrak I_0),v(\mathfrak J_1),\ldots,v(\mathfrak J_L)\},
\]
so some unused library ideal improves it.  Conversely, if step (c) occurs,
the same minimum remains strictly larger than $v(\mathfrak I_\star)$, proving
that the full library sum cannot have the target integral closure.  Finiteness
of $\mathcal L$ gives termination.

For positive-semidefinite forms,
$\ker(A+B)=\ker A\cap\ker B$.  Thus a nonzero vector in
\eqref{en:eq:fisher-library-kernel} survives every nonnegative shot mixture.
If the intersection is zero and the current kernel $K$ is nonzero, some
$G_\ell$ is nonzero on a vector of $K$, so
$\dim(K\cap\ker G_\ell)<\dim K$.  Iteration proves the bound and the claimed
positive definiteness on $U$.
\end{proof}

\begin{remark}[What is new and what is inherited]
The valuative criterion, finiteness of Rees valuations, and the
positive-semidefinite kernel identity are established ingredients.  The new
statement is their operational completion-or-refusal rule for a declared quantum
measurement--record library: it makes ``augment'' a terminating outcome, and
distinguishes failure of the current representation from impossibility under
the available hardware controls.  It does not claim that an arbitrary
laboratory library satisfies \eqref{en:eq:library-completion}; failure of that
condition is precisely the certified refusal branch.
\end{remark}

\begin{remark}[Assumption and output checklist]
The theorem requires (i) analytic response germs in an analytically unramified
Noetherian local domain, (ii) an explicitly fixed target ideal, (iii) a finite
predeclared library whose elements are physically executable and contribute
ideals inside that target, and (iv) computable strict witnesses for the chosen
representation. A Rees witness becomes a laboratory direction only after the
real-accessibility and admissibility tests stated above. Under these conditions
the theorem guarantees completion or failure \emph{relative to the
library}; it neither searches over all quantum measurements nor asserts that
an arbitrary empirical QPU family has an analytic response model. Users who
seek only declared finite comparisons stop after the Hellinger audit and do
not need these assumptions.
\end{remark}

\begin{example}[Normalized blow-up to an executable Qiskit control arc]
\label{en:ex:qiskit-normalized-blowup}
The running circuit also gives a complete, nonprincipal example of the
witness-lifting procedure.  Retain its full computational response, but test
the analytic certificate
\begin{equation}\label{en:eq:blowup-certificate}
 C_\star(a,b)=\bigl(p_{010}(a,b),p_{001}(a,b)^2\bigr).
\end{equation}
The square is deliberate: it models a nonlinear reported quantity and creates
a genuinely weighted audit.  It is estimable without asymptotic bias from
counts by replacing $p_{001}^2$ with
$N_{001}(N_{001}-1)/[n(n-1)]$ when $n\ge2$.
At $(0,0)$, units and higher-order terms do not change the ideals, so
\begin{equation}\label{en:eq:blowup-ideal-pair}
 \mathfrak I_{C_\star}=(a^2,b^4),\qquad
 \mathfrak I_P=(a^2,b^2).
\end{equation}
The compact Newton facet of $(a^2,b^4)$ is $2i+j=4$; its primitive inward
normal gives the Rees valuation
\begin{equation}\label{en:eq:blowup-valuation}
 v(a)=2,\quad v(b)=1,\quad
 v(\mathfrak I_{C_\star})=4>2=v(\mathfrak I_P).
\end{equation}
In the corresponding real weighted chart, the normalized blow-down has the
form
\begin{equation}\label{en:eq:weighted-blowdown}
 a=s^2u,\quad b=sv,\quad E=\{s=0\}.
\end{equation}
The point $(s,u,v)=(0,1,1)$ is real and smooth.  The transverse chart arc
$s=t$, $u=v=1$ descends to the directly programmable control curve
\begin{equation}\label{en:eq:executable-rees-arc}
 \gamma(t)=(a(t),b(t))=(t^2,t).
\end{equation}
Along it,
\begin{align}\label{en:eq:executable-orders}
 p_{010}(\gamma(t))&=\sin^2(t^2)=t^4+O(t^8),\nonumber\\
 p_{001}(\gamma(t))&=\cos^2(t^2)\sin^2t=t^2+O(t^4),\\
 p_{001}(\gamma(t))^2&=t^4+O(t^6),\nonumber
\end{align}
and therefore $\nu_{C_\star}(\gamma)=4>2=\nu_P(\gamma)$ exactly as the
exceptional divisor predicts.  A Qiskit implementation with two
number-preserving \texttt{XXPlusYYGate} rotations verifies the dyadic orders:
at $t=0.025$, the successive-ratio estimates are $3.999997$ for $p_{010}$,
$3.998180$ for $p_{001}^2$, and $1.999090$ for $p_{001}$.

This example also fixes the division of labor.  Qiskit constructs the circuit,
evaluates or samples its response, and executes the descended controls.  The
Newton facet, normalization, and Rees valuation are symbolic commutative-
algebra calculations; Qiskit alone does not compute a normalized blow-up.  The
complete reproducible case therefore combines symbolic ideal computation with
Qiskit response verification rather than presenting either as a substitute for
the other.  The certificate is intentionally a failure case: its nonlinear
second coordinate delays a physically visible $b$ response, and the audit
returns the curve on which that delay can be tested.
\end{example}

Finite data require deferral as a formal outcome.  Let $\mathcal K$ index
the predeclared comparisons and let $(p_k,q_k)\in\Delta_{m_k}\times
\Delta_{m_k}$ be the two unknown response laws in comparison $k$.  Let
$\mathcal A\subseteq\mathcal K\times\mathcal T$ be the entire audited family of
comparison--record claims, where each $T\in\mathcal T$ is deterministic.  For
$a=(k,T)$ define
\begin{align}\label{en:eq:compact-pair-functionals}
 B_a(p_k,q_k)&=H^2(T_\#p_k,T_\#q_k),\nonumber\\
 D_a(p_k,q_k)&=H^2(p_k,q_k),\nonumber\\
 R_a(p_k,q_k)&=D_a(p_k,q_k)-B_a(p_k,q_k)\ge0.
\end{align}
By Proposition~\ref{en:prop:markov-kernel}, the same definitions and decisions
hold for any predeclared fixed Markov kernel after replacing $T_\#$ by $K$.
Let
$\mathcal C_N\subseteq\prod_{k\in\mathcal K}
(\Delta_{m_k}\times\Delta_{m_k})$ satisfy
\begin{equation}\label{en:eq:compact-joint-coverage}
 \Pr\{((p_k,q_k))_{k\in\mathcal K}\in\mathcal C_N\}\ge1-\alpha.
\end{equation}
Thus coverage is simultaneous across both distributions, all coordinates, and
all comparisons from which any claim in $\mathcal A$ will be selected.  Fix a
task-discrepancy budget $\tau_a$ and a distinct record-loss budget $\rho_a$.
The two budgets answer different questions: $D_a$ tests the circuit or task
response against its reference, whereas $R_a$ tests whether $T$ preserves that
response.  A large, scientifically intended response may have large $D_a$ and
zero $R_a$; it must not be rejected as a representation failure.
\begin{theorem}[Simultaneous finite-data decisions]
\label{en:thm:three-way}
For every $a=(k,T)\in\mathcal A$, define two three-way decisions on the same
joint confidence region:
\begin{align*}
 &\text{record approve: }\sup_{\mathcal C_N}R_a\le\rho_a,\\
 &\text{record refuse: }\inf_{\mathcal C_N}R_a>\rho_a,\\
 &\text{task pass: }\sup_{\mathcal C_N}D_a\le\tau_a,\\
 &\text{task fail: }\inf_{\mathcal C_N}D_a>\tau_a.
\end{align*}
Defer the corresponding decision in every remaining case, and retain the full
histogram whenever the record decision is refuse or defer.  A combined
``task within tolerance under an adequate record'' claim is accepted only when
both task pass and record approve hold.  With probability at least $1-\alpha$,
all task and record decisions over $\mathcal A$ are simultaneously supported;
the familywise probability of any unsupported nondeferred decision is at most
$\alpha$.  In particular, $D_a>\tau_a$ alone never refuses the record, and
$R_a>\rho_a$ alone does not assert that the underlying task failed.  A support
floor $1-\sqrt{1-\ell}$ may certify task failure when its lower confidence
bound exceeds $\tau_a$, but cannot by itself approve a record.
\end{theorem}
\begin{proof}
On the joint coverage event the true pair for every $k$ belongs to the same
region over which all extrema are taken.  Each upper-bound decision proves its
own budget and each lower-bound decision proves violation of its own budget,
simultaneously for every $a\in\mathcal A$.  The combined claim is the
intersection of the two supported upper-bound events.  Any unsupported
nondeferred familywise decision is contained in the complement of the joint
coverage event.
\end{proof}

\begin{remark}[Absolute loss control and nonvacuous relative capture]
\label{en:rem:nonvacuous-capture}
Record approval in Theorem~\ref{en:thm:three-way} certifies the absolute statement
$R_a\le\rho_a$.  This statement remains valid when $D_a=0$, but then means
only that the tested pair contains no response for the record to hide.  It
does not certify a nontrivial capture fraction.  A claim
$B_a/D_a\ge 1-\varepsilon_a$ therefore requires a predeclared detectability
level $\kappa_a>0$ and is certified only if
\begin{equation}\label{en:eq:relative-capture-gate}
 \inf_{\mathcal C_N}D_a\ge\kappa_a,
 \qquad
 \sup_{\mathcal C_N}\{R_a-\varepsilon_aD_a\}\le0.
\end{equation}
It is rejected if
$\inf_{\mathcal C_N}\{R_a-\varepsilon_aD_a\}>0$ and the detectability gate
is certified; otherwise it is deferred.  Thus no ratio is divided by a
quantity statistically compatible with zero.  The absolute-loss decision is
the default; relative capture is an optional, stronger claim.
\end{remark}

For an equivalence target $E_\delta=[-\delta,\delta]$ and a simultaneous
confidence interval $C$, the same joint-region logic gives
\begin{equation}\label{en:eq:equivalence-three-way}
 \begin{aligned}
 C\subseteq E_\delta&\Rightarrow\text{accept equivalence},\\
 C\cap E_\delta=\varnothing&\Rightarrow\text{reject equivalence}.
 \end{aligned}
\end{equation}
with defer in every crossing case.  Failure to accept equivalence is therefore
not evidence of non-equivalence.  This distinction is used in the first IBM
phase experiment below.

\begin{corollary}[Independent defer-and-resolve follow-up]
\label{en:cor:independent-followup}
Let $D_1$ be an initial data set and let a follow-up design
$g(D_1)$---including its estimand, comparison family, confidence construction,
and stopping rule---be frozen before collecting future data $D_2$.  Suppose
that, conditionally on $D_1$, the new data are independent of $D_1$ under the
declared second-stage experiment and its confidence region satisfies
\begin{equation}\label{en:eq:conditional-followup-coverage}
 \Pr\{\xi_2\in C_2(D_2;D_1)\mid D_1\}\ge1-\alpha_2.
\end{equation}
Then the second-stage accept--reject--defer decision has unconditional error at
most $\alpha_2$.  If the first-stage and second-stage statements are reported
jointly, their simultaneous coverage is at least
$1-(\alpha_1+\alpha_2)$.  A second-stage acceptance concerns the declared
future-data estimand $\xi_2$ and does not retroactively replace a first-stage
defer decision or establish equality of hardware responses across epochs.
\end{corollary}
\begin{proof}
Take expectation of
\cref{en:eq:conditional-followup-coverage} with respect to $D_1$.  The joint
statement follows from a union bound over the two coverage events.  The final
claim is a scope statement: the two stages index distinct frozen experiments,
so coverage of the second estimand cannot alter the recorded first-stage event.
\end{proof}

Suppose comparison $k$ uses independent multinomial histograms
$\widehat p_k,\widehat q_k$ with sample sizes $N_{k,p},N_{k,q}$.  Put
$M=2\sum_{k\in\mathcal K}m_k$ and
\begin{equation}\label{en:eq:compact-confidence-box}
 \delta_{k,s}=\sqrt{\frac{\log(2M/\alpha)}{2N_{k,s}}},\qquad s\in\{p,q\}.
\end{equation}
Let $\mathcal C_N$ contain precisely the pairs satisfying
\begin{equation}\label{en:eq:compact-confidence-product}
 \begin{aligned}
 |p_k(y)-\widehat p_k(y)|&\le\delta_{k,p},\\
 |q_k(y)-\widehat q_k(y)|&\le\delta_{k,q},
 \qquad \text{for all }k,y.
 \end{aligned}
\end{equation}
Hoeffding's inequality and one union bound over the $M$ scalar coordinates give
joint coverage at least $1-\alpha$.  For $K$ equal-alphabet comparisons with
$m$ outcomes and equal sample size $N$, this reduces to
$\delta=\sqrt{\log(4Km/\alpha)/(2N)}$.  Optimizing
\cref{en:eq:compact-pair-functionals} over this product box is conservative;
sharper joint multinomial regions or a predeclared allocation of $\alpha$
across disjoint families may replace it.  Intervals constructed separately
after selecting the most favorable record do not provide the familywise
guarantee.  A point estimate alone is never an acceptance certificate.

There is a finite-sample sharpening that changes no decision logic.  For each
histogram, invert a predeclared multinomial likelihood-level covering collection
to obtain an exact region $\mathcal C^\star_{N,\alpha}$, and take the product
over the preallocated comparison levels.  Such regions have finite-sample
coverage for any finite alphabet and can have minimum average volume among
exact multinomial regions~\cite{ChafaiConcordet2009,MalloyTripathyNowak2020}.
Replacing the coordinate box by this product region in Theorem~\ref{en:thm:three-way} therefore
preserves familywise validity while weakly reducing defer whenever the new
region is contained in the old one.  For large alphabets, relative-entropy
concentration supplies a computational outer approximation whose natural scale
is governed by $(m-1)/N$ rather than by $m$ separate coordinate decisions
\cite{Agrawal2020}.

The gain can be stated without an asymptotic approximation.  Suppose a joint
region guarantees $H(p,\widehat p)\le r_p$ and
$H(q,\widehat q)\le r_q$.  The triangle inequality and data processing imply
\begin{equation}\label{en:eq:hellinger-functional-width}
 \begin{aligned}
 \left|H^2(p,q)-H^2(\widehat p,\widehat q)\right|
   &\le 2(r_p+r_q),\\
 |R_T(p,q)-R_T(\widehat p,\widehat q)|&\le4(r_p+r_q).
 \end{aligned}
\end{equation}
Consequently a comparison cannot remain deferred once its population total and
residual lie farther than these certified widths from every relevant decision
boundary.  Equation~\eqref{en:eq:hellinger-functional-width} is an explicit
power criterion: additional shots are required only until the confidence
radii fall below the smallest predeclared margin, rather than until every
multinomial coordinate is estimated to a common absolute precision.

\subsection{Running-example corollary and experimental predictions}
\begin{proposition}[Three-qubit refusal-and-repair certificate]
\label{en:ex:three-qubit-ideal-audit}
\label{en:prop:ibm-mechanism-predictions}
For the circuit of Section~2, $T_1=n_1$ has
$F_{T_1}=\operatorname{diag}(4,0)$ and a positive $b$-axis Hellinger residual,
so it must be rejected.  The augmentation $T_2=(L,n_1,n_2)$ is injective on
the allowed support and is accepted for hardware comparisons only when the
confidence bound for $R_{T_2}^H$ meets the declared tolerance.  At the boundary,
the computational histogram generates $(a^2,b^2)$ while the signed
interference witnesses generate $(a,b)$.  Thus the predeclared operational
predictions are: refusal of $T_1$ on the $b$ axis, conditional
measurement-relative repair by $T_2$, and recovery of odd signed-amplitude
responses by interference settings.  The relative-phase extension in
Section~2 predicts an even $X_{02}$ response, an odd $Y_{02}$ response, and
$Z$-basis invariance.  Section~4 tests these predictions on one executable real
phase arc and then applies Corollary~\ref{en:cor:independent-followup} to a
separately frozen $Z$ replication.
\end{proposition}
In response-ideal notation the calculation is
\begin{equation}\label{en:eq:compact-running-ideals}
 \begin{aligned}
 \mathfrak I_\rho&=(a,b),& \mathfrak I_P&=(a^2,b^2),\\
 \mathfrak I_{T_1}&=(a^2),& \mathfrak I_{T_2}&=(a^2,b^2).
 \end{aligned}
\end{equation}
Thus the $b$ axis is simultaneously a positive Fisher-loss eigenvector, a
positive finite Hellinger residual, and an unequal-valuation witness.  This
agreement connects local, finite, and all-direction diagnostics without
promoting any one scalar residual to a universal certificate.
These predictions instantiate the three main theorems without claiming that a
finite set of hardware paths proves integral-closure equality or device QFI.

\section{Hardware tests of the decision rules}
The experiments test observable consequences of Section~3, not its all-arc
algebraic equivalences. IBM Kingston provides a prospectively frozen
test of record refusal, approval after augmentation, and targeted witness
recovery; a second phase sweep and independent replication implement defer and
future-data resolution.  IQM Garnet provides a larger-alphabet, two-epoch
stress test under device drift.  None of these experiments estimates device QFI, proves
integral-closure equality, or establishes hardware advantage.

\subsection{Prospectively frozen IBM test}
\label{en:sec:ibm-mechanism}
Two $XX+YY$ Givens rotations implemented the three-qubit circuit of Section~2.
The computational setting compared $T_1=n_1$ with
$T_2=(L,n_1,n_2)$, where $L=\mathbf1\{n_0+n_1+n_2\ne1\}$.  Two additional
number-preserving interference settings measured the signed $X_{01}$ and
$X_{02}$ responses.  Before QPU submission, the protocol fixed the control
axes and radii, two blocks, 4,096 shots per configuration, layout, seed,
tolerances, and disabled mitigation.  The resulting 114 configurations used
466,944 shots on \texttt{ibm\_kingston}.  Timestamps, job identifiers,
manifest and analysis hashes, and the circuit audit are reported in
the supplementary protocol ledger.

On the predeclared $b$ axis at $h=0.20$, the mean full-histogram displacement
was $0.023025$, whereas $T_1$ retained $1.56\times10^{-5}$.  Its capture ratio
was $0.0676\%$, and the missed residual was
\begin{equation}\label{en:eq:compact-ibm-refusal}
 \begin{aligned}
 \overline{R^H_{T_1}}&=0.023009,\\
 95\%\ \text{bootstrap interval}&=[0.021111,0.025912].
 \end{aligned}
\end{equation}
None of 5,000 multinomial bootstrap replicates gave a nonpositive mean
residual.  The frozen rule therefore refused $T_1$ rather than reporting the
only proposed scalar record as adequate.

For $T_2$, the empirical full-minus-quotient residual was zero at numerical
precision on the frozen comparisons.  The only unresolved three-bit fiber was
$\{110,111\}$.  Combining its baseline one-sided mass bound with its largest
primary-condition mass gave the conservative Hellinger-gap bound
\begin{equation}\label{en:eq:compact-ibm-repair}
 2.99\times10^{-4}<0.005.
\end{equation}
Thus $T_2$ passed only the predeclared measurement-relative tolerance.  This
is conditional acceptance on the tested comparisons, not universal
sufficiency.

The interference settings separately tested the signed directions hidden at
the computational boundary.  Across the two blocks,
$O_{X_{01}}=0.372803,0.361328$ and
$O_{X_{02}}=0.384766,0.387939$; simultaneous intervals excluded zero.
Matched computational-basis odd contrasts remained within the frozen
$[-0.05,0.05]$ equivalence margin, while their even responses were nonzero.
Thus the data instantiate
\begin{equation}\label{en:eq:compact-ibm-chain}
 \mathfrak I_Q\subsetneq\mathfrak I_C\simeq\mathfrak I_P
 \subsetneq\mathfrak I_\rho
\end{equation}
only on the frozen finite paths.  The symbol $\simeq$ denotes the tolerance
statement.

\subsection{Executable phase direction and global withholding}
The measurement-loss prediction was tested in a second prospectively frozen
experiment using
\begin{equation}\label{en:eq:phase-hardware-state}
 \begin{aligned}
 |\psi_\phi\rangle&=\frac{|100\rangle+e^{i\phi}|001\rangle}{\sqrt2},\\
 \phi&\in\{\pm.05,\pm.10,\pm.20,\pm.35\}.
 \end{aligned}
\end{equation}
For $D_{02}=n_0-n_2$, the ideal terminal responses are
\begin{equation}\label{en:eq:phase-response-signatures}
 D_{02}^{Z}(\phi)=0,\qquad
 D_{02}^{X_{02}}(\phi)=\cos\phi,\qquad
 D_{02}^{Y_{02}}(\phi)=\sin\phi.
\end{equation}
Thus $Z$ is invariant, $X_{02}$ is even, and $Y_{02}$ supplies the odd witness
that computational counts lack.  The protocol fixed two randomized blocks,
33 configurations per block, 4,096 shots per configuration, the same physical
layout $(141,140,142)$, disabled mitigation, and a maximum-deviation
multinomial bootstrap.  The three predeclared families received error budgets
$0.015$, $0.015$, and $0.020$, summing to familywise level $0.05$.

The $X_{02}$ odd contrasts met the equivalence rule, its curvature at
$|\phi|=.35$ exceeded the frozen threshold in both blocks, and every
$Y_{02}$ odd-response interval excluded zero.  The conjunctive protocol did
not, however, certify its complete phase claim.  For the $Z$ comparison in
Block B at $h=.35$,
\begin{equation}\label{en:eq:first-phase-defer}
 \widehat\Delta^{Z}_{B,.35}=0.03040,\qquad
 C^{Z}_{B,.35}=[-0.00354,0.06433].
\end{equation}
The point estimate lies in the frozen equivalence region $[-.05,.05]$, but the
simultaneous interval is not contained in it.  By
\cref{en:eq:equivalence-three-way}, this is defer rather than evidence of
non-equivalence.  We therefore report component-level support for the even and
odd interference responses and withhold the conjunctive phase claim.

\subsection{Independent replication of the deferred $Z$ comparison}
After recording the defer decision, a distinct follow-up was frozen before any
new QPU data were collected.  It targeted only $Z$ at $h=.35$ and retained the
same backend and physical triple.  Four independently submitted blocks each
contained forward states from \cref{en:eq:phase-hardware-state}, reverse states
globally equivalent to the forward path at $-\phi$, and pre/post $\phi=0$
anchors.  The design used eight configurations per block and 12,288 shots per
configuration, for 393,216 shots in total.  Its fixed stopping rule prohibited
optional continuation.

The forward odd contrasts in Blocks A--D were
\begin{equation}\label{en:eq:z-replication-blocks}
 -0.00586,\quad 0.00553,\quad -0.00269,\quad -0.00834,
\end{equation}
and the pooled result was
\begin{equation}\label{en:eq:z-replication-pooled}
 \widehat\Delta_{F,\mathrm{pool}}=-0.00284,\qquad
 C_{F,\mathrm{pool}}=[-0.02058,0.01490].
\end{equation}
Every simultaneous forward-block interval lay inside $[-.05,.05]$; the
predeclared positive-reproduction criterion was false; all anchor-drift
intervals satisfied the same practical margin.  Forward--reverse sums also met
their equivalence rule, but because the response itself was null this does not
identify an orientation-reversing microscopic mechanism.  The proper
classification is that the earlier positive deviation was not reproduced
beyond the declared margin.

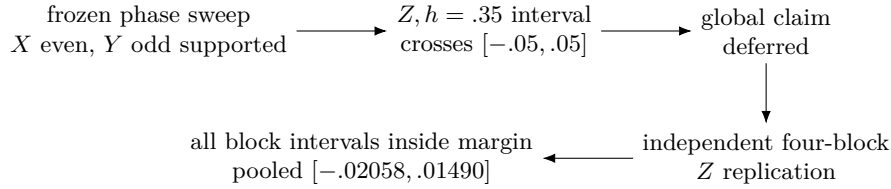
\begin{figure*}[t]
\centering
\begin{tikzpicture}[>=Latex,node distance=8mm and 12mm,
 every node/.style={align=center,font=\small}]
\node (p) {frozen phase sweep\\$X$ even, $Y$ odd supported};
\node[right=of p] (z) {$Z,h=.35$ interval\\crosses $[-.05,.05]$};
\node[right=of z] (d) {global claim\\deferred};
\node[below=of d] (r) {independent four-block\\$Z$ replication};
\node[left=of r] (a) {all block intervals inside margin\\pooled $[-.02058,.01490]$};
\draw[->] (p)--(z); \draw[->] (z)--(d); \draw[->] (d)--(r); \draw[->] (r)--(a);
\end{tikzpicture}
\caption{The audit--defer--resolve sequence.  The replication supplies a new
future-data statement under Corollary~\ref{en:cor:independent-followup}; it does
not rewrite the first experiment's deferred decision.}
\label{en:fig:phase-defer-resolve}
\end{figure*}
\FloatBarrier

Together the two studies support the predicted phase-response components and
demonstrate the refusal sequence. They do not constitute one
post-hoc all-pass protocol, identify a microscopic noise channel, certify a
Rees valuation, or establish generality across layouts or
devices.

\subsection{Protocol-preserving replication on Marrakesh and Kingston}
\label{en:sec:cross-device-replication}
We next asked whether the finite decision pattern survives a change of
processor and physical layout.  The logical circuits, control points,
randomization seed, two-block structure, 4,096 shots per configuration,
disabled mitigation, records, tolerances, and analysis code were held fixed.
For each processor the layout alone was selected from the contemporaneous
calibration data.  The three-qubit layouts were $(147,148,149)$ on
\texttt{ibm\_marrakesh} and $(90,91,98)$ on \texttt{ibm\_kingston}; the
five-qubit paths were $(5,6,7,8,9)$ and $(88,89,90,91,98)$, respectively.  The Kingston substitution was
specified after the original Marrakesh--Fez registration but before Kingston
counts were examined.  We therefore call it a protocol-preserving independent
replication, not a completed preregistered two-device test.

For the three-qubit family, each device and block compared $b=0$ with
$b=\pm0.20$ and tested the phase pair $\phi=\pm0.20$ in the $Z$, $X_{02}$,
and $Y_{02}$ settings.  In all eight device--block--sign comparisons the
Bonferroni simultaneous lower bound for the $T_1=n_1$ Hellinger residual was
positive, so the scalar record was refused.  The $Z$ odd contrasts were
equivalent to zero, $X_{02}$ remained even, and every simultaneous $Y_{02}$
interval excluded zero.  On Kingston the full $b$-response distances ranged
from $0.01965$ to $0.02458$, while the $T_2=(L,n_1,n_2)$ point residuals were
only $9.97\times10^{-5}$ to $5.99\times10^{-4}$.  Nevertheless, the
simultaneous upper bound crossed the $0.005$ tolerance in at least one tested
unit.  The correct joint decision is therefore defer, not repair.

The five-qubit circuit is a sequential open-control $R_y$ preparation on the
path $P_5$, with reflected controls
$(t+\delta,t,t,t,t-\delta)$ and
$(t-\delta,t,t,t,t+\delta)$.  At the primary radius $|\delta|=0.24$, the two
Kingston blocks gave
\begin{equation}\label{en:eq:cross-device-five-contrast}
 \begin{aligned}
 H^2_{\rm full}&=0.05154,\ 0.05444,\\
 H^2_{\rm weight}&=0.000680,\ 0.000480.
 \end{aligned}
\end{equation}
Marrakesh gave the same qualitative separation.  Across both devices and both
blocks, the simultaneous lower bounds established a detectable full response
and a positive Hamming-weight residual.  Hamming weight was therefore refused.
The first-moment and richer support--moment records removed most of the point
residual, but their simultaneous intervals crossed the $0.01$ acceptance
threshold; both decisions remain defer.

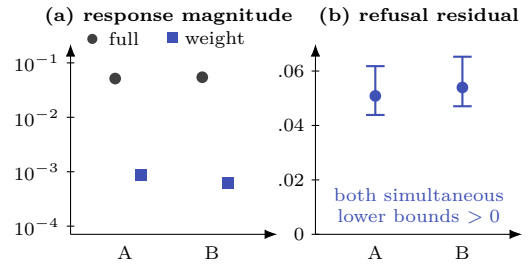
\begin{figure}[H]
\centering
\begin{tikzpicture}[x=.72cm,y=.72cm,>=Latex,font=\scriptsize]
\node[font=\scriptsize\bfseries] at (2.45,4.55) {(a) response magnitude};
\draw[->] (.65,.55)--(.65,4.05);
\draw[->] (.65,.55)--(4.45,.55);
\foreach \y/\lab in {.70/$10^{-4}$,1.70/$10^{-3}$,2.70/$10^{-2}$,3.70/$10^{-1}$}
 {\draw (.58,\y)--(.72,\y) node[left=2pt] {\lab};}
\node[below] at (1.65,.50) {A}; \node[below] at (3.25,.50) {B};
\fill[black!75] (1.47,3.412) circle (2.2pt);
\fill[black!75] (3.07,3.436) circle (2.2pt);
\fill[designblue] (1.83,1.533) rectangle +(4.4pt,4.4pt);
\fill[designblue] (3.43,1.381) rectangle +(4.4pt,4.4pt);
\fill[black!75] (1.05,4.14) circle (1.9pt); \node[anchor=west] at (1.18,4.14) {full};
\fill[designblue] (2.43,4.08) rectangle +(3.8pt,3.8pt);
\node[anchor=west] at (2.58,4.14) {weight};
\node[font=\scriptsize\bfseries] at (7.05,4.55) {(b) refusal residual};
\draw[->] (5.15,.55)--(5.15,4.05);
\draw[->] (5.15,.55)--(8.95,.55);
\foreach \y/\lab in {.55/0,1.55/.02,2.55/.04,3.55/.06}
 {\draw (5.08,\y)--(5.22,\y) node[left=2pt] {\lab};}
\node[below] at (6.25,.50) {A}; \node[below] at (7.85,.50) {B};
\draw[designblue,thick] (6.25,2.743)--(6.25,3.640);
\draw[designblue,thick] (6.08,2.743)--(6.42,2.743);
\draw[designblue,thick] (6.08,3.640)--(6.42,3.640);
\fill[designblue] (6.25,3.093) circle (2.2pt);
\draw[designblue,thick] (7.85,2.903)--(7.85,3.813);
\draw[designblue,thick] (7.68,2.903)--(8.02,2.903);
\draw[designblue,thick] (7.68,3.813)--(8.02,3.813);
\fill[designblue] (7.85,3.248) circle (2.2pt);
\node[align=center,text=designblue] at (7.05,1.08) {both simultaneous\\lower bounds $>0$};
\end{tikzpicture}
\caption{Kingston five-qubit primary comparison at $|\delta|=.24$.
(a) The log scale displays the full and Hamming-weight Hellinger displacements
for Blocks A and B.  (b) Points are
$R^H_{\rm weight}=H^2_{\rm full}-H^2_{\rm weight}$; intervals are the
Bonferroni-adjusted simultaneous 95\% bootstrap intervals across the
predeclared device--block family.  Their positive lower endpoints determine
refusal of the Hamming-weight record.}
\label{en:fig:five-qubit-refusal}
\end{figure}

The simultaneous intervals use Bonferroni-adjusted percentile bootstrap tails
within each predeclared claim family across devices and blocks.  This is more
conservative than the comparison-wise intervals and is the basis of every
joint statement in this subsection.  The result is not equality of device
responses: it is replication of the decision pattern under the declared
tolerances.

\subsection{Predeclared completion test across circuit classes}
\label{en:sec:operational-completion-hardware}
The preceding replication asks whether a proposed record should be refused.
We next froze a stronger experiment before submission: can the same audit
distinguish an inadequate record, an informative but still incomplete
augmentation, and a measurement-relative completion on two structurally
different four-qubit families?  The protocol used two independently randomized
blocks on each of \texttt{ibm\_marrakesh} and \texttt{ibm\_kingston}, with 20
configurations and 4,096 shots per block.  Mitigation, gate twirling,
measurement twirling, and dynamical decoupling were disabled.  Only the
calibration-dependent physical paths, $(146,147,148,149)$ and
$(89,90,91,98)$, were selected at submission.

The phase family was
\begin{equation}\label{en:eq:four-ghz-family}
 \begin{aligned}
 |\mathrm{GHZ}_\phi\rangle
   &=\frac{|0000\rangle+e^{i\phi}|1111\rangle}{\sqrt2},\\
 \phi&\in\{-.35,-.20,0,.20,.35\}.
 \end{aligned}
\end{equation}
Computational measurement was predicted to be phase invariant, $X$ parity to
be even, and $Y$ parity to carry a nonzero odd response.  Across both devices
and both blocks, every simultaneous $Z$ interval met the equivalence rule and
every simultaneous $Y$-odd interval excluded zero.  At $|\phi|=.35$, the
observed $Y$-parity Hellinger distances were $0.0371$--$0.0386$ on Marrakesh
and $0.0458$--$0.0462$ on Kingston; the corresponding full-minus-parity
residuals were at most $0.00140$.  Parity was therefore approved for the
measured phase response.

The preregistered $X$-even conjunction did not pass.  Marrakesh showed positive
odd contrasts in both blocks at both radii; Kingston showed a smaller,
block-dependent odd component.  This is neither removed as an outlier nor
reinterpreted as the ideal prediction.  It is a negative control demonstrating
that the audit can reject a planned symmetry statement while approving a
different record-relative claim on the same data.  The experiment does not
identify whether the asymmetry originates in coherent gate error, basis-change
compilation, readout, or epoch-specific drift.

The second family preserves two particles while redistributing their
occupations.  For reflected controls at radii $.15$ and $.30$, every
device--block unit detected a full-histogram response.  At radius $.30$,
\begin{equation}\label{en:eq:four-number-ladder}
 \begin{aligned}
 H^2_{\rm full}&=0.2165\text{--}0.2314,\\
 R^H_{\rm weight}&=0.2157\text{--}0.2306,\\
 R^H_{\rm right}&=0.0568\text{--}0.0665,\\
 R^H_{\rm resolved}&=0.
 \end{aligned}
\end{equation}
Thus Hamming weight is refused, adding right occupation is informative but
still refused, and the resolved number record is approved.  This is the
hardware realization of the finite-library logic in
Theorem~\ref{en:thm:accessible-completion}: successive augmentation is not
called completion merely because it improves the point estimate; it must
remove the declared residual within simultaneous uncertainty.  The zero
resolved residual is a measurement-relative statement for this finite record,
not a claim of state informational completeness or dimension-independent
compression.

\subsection{IQM Garnet stress test under drift}
The IQM study asks whether residual-based compression remains informative when
the output alphabet and temporal variation are larger.  SDT QuREKA provided
managed access to IQM Garnet under the TXST--SDT collaboration
\cite{IQMGarnet2024}.  The frozen design used $n\in\{5,8,12\}$, 36 OpenQASM
circuits per epoch, 10,000 shots per setting, and two epochs, for 72 jobs and
720,000 shots.  Opaque compilation precludes physical-edge and native-gate-cost
claims.

Absolute errors changed between epochs, but circuit difficulty rankings were
stable.  For $H^2$, TV, and support leakage, the cross-epoch Pearson
correlations were $0.9496$, $0.9506$, and $0.9664$, while the paired mean
changes were $-0.03322$, $-0.06201$, and $-0.03732$.  The structured circuits
improved over the matched baseline at $n=8$ in both epochs but not at $n=12$;
this repeated boundary rules out a universal scale advantage.

Two results connect directly to refusal.  A structured $n=12$ setting retained
the non-toric residual
\begin{equation}\label{en:eq:compact-iqm-refusal}
 R_{\rm nt}=0.06992,
\end{equation}
so the quotient improved one prediction without concealing its dominant
failure.  Full drift was significant in all 36 settings and quotient drift in
35, but the median attenuation ratio was
\begin{equation}\label{en:eq:compact-iqm-drift}
 \operatorname{median}_s\frac{D_s^T}{D_s^{\rm full}}=0.07972.
\end{equation}
The quotient therefore suppressed roughly $92\%$ of measured drift while
remaining demonstrably noninvariant.

The coefficient-free barycentric predictor was formed post hoc in Epoch A and
then applied without refitting to the corresponding Epoch-E1 histograms.  It
improved all six reported settings in each epoch and reduced Epoch-E1 MAE from
$0.0089916$ to $0.0073886$.  Because Epoch E1 was not prospectively reserved
for this hypothesis, this is retrospective temporal replication, not a
confirmatory performance claim.  Detailed means, resampling intervals,
chord-Fisher diagnostics, and the low-shot bias--variance crossover remain in
the reproducibility appendix.

\subsection{What the experiments do and do not validate}
The experiments test only the observable decisions attached to the
four-theorem chain: detection of hidden response, targeted augmentation,
finite-data refusal or deferral, and independently checkable resolution.  They
are not performance benchmarks and do not empirically prove the all-arc
integral-closure statements.
The studies have distinct and noninterchangeable evidential roles.  The first
IBM study is a prospective within-device test on one circuit family,
and the four-block follow-up is an independent temporal replication of one
deferred comparison on the same backend and layout.  The
Marrakesh--Kingston study tests whether the same decision pattern survives a
processor and layout change under a recorded backend amendment.  The
four-qubit study tests the augmentation ladder on two additional circuit
classes. IQM is a heterogeneous larger-alphabet stress test on a
different device and circuit collection; its predictor analysis remains
retrospective.
The combined evidence supports the predicted $b$-direction refusal of $T_1$,
conditional approval of the recorded histogram after augmenting it to $T_2$,
and recovery of a
signed-amplitude direction.  The phase study supports the even/odd
interference signatures but withholds its global claim when $Z$ equivalence is
unresolved; the independent replication subsequently certifies practical $Z$
equivalence on new blocks without altering the original decision.  The
four-qubit study shows that the same rules can distinguish successive record
augmentations. IQM supplies a larger-alphabet example in which drift
is attenuated but not removed and a large residual blocks a low-dimensional
interpretation.  These experiments do not enumerate all
analytic arcs or Rees valuations, identify a microscopic noise cause, or turn
conditional acceptance into a device-independent sufficient statistic. Across
the studies, the QPU counts determine a residual and a decision for each
specified augmentation; deferred comparisons remain recorded and are resolved
only with independently collected data. The replicated Hamming-weight failure expands the tested circuit
class, but two Heron processors do not establish cross-device universality.

\section{Discussion and conclusion}
\subsection{Relation to prior work and claim boundary}
Toric ideals and exponential-family geometry already describe structured
discrete models~\cite{GeigerMeekSturmfels2006,MontufarRauhAy2014,Molitor2021,
Molitor2025}.  Quantum-circuit work supplies coherent graphical-model samplers,
constraint-preserving dynamics, and fixed-weight encoders
~\cite{PiatkowskiZoufal2024,Hadfield2019,Monbroussou2023,Farias2024}.
Quantum Fisher information and symmetry verification address state sensitivity
and sector-based error control~\cite{Meyer2021,Vitale2024,BonetMonroig2018}.
These results motivate the circuit families and provide several calculations
used here.  They do not imply that a selected low-dimensional record preserves
an empirical hardware response.

The attainable-metric fibre is closest in spirit to preprocessing-optimized
Fisher information and randomized multiparameter measurement design
~\cite{ZhouMichalakisGefen2023,ZhouChen2026}.  Those works optimize a protocol
over a specified admissible class.  We instead retain the whole certified
inner set $\mathcal G_{\rm att}$ because different control directions may
require incompatible settings, and use it only to decide whether a proposed
repair is executable.  We do not claim their near-optimality or optimization
results.  Likewise, finite-sample quantum metrology shows why large QFI alone
does not certify finite-data performance~\cite{MeyerKhatriFranca2025}; our
joint confidence rule addresses the narrower problem of auditing observed
full and compressed distributions, not general parameter-estimation risk.

The algebraic criteria used by Theorems~\ref{en:thm:two-level-completeness}
and~\ref{en:thm:finite-rees-audit} are also established.
Integral closure, the analytic-arc criterion, Rees valuations, and the
\L{}ojasiewicz exponent belong to classical commutative and analytic geometry
~\cite{LejeuneTeissier2008,HunekeSwanson2006}.  Theorem~\ref{en:thm:two-level-completeness} applies the arc
criterion to two adjacent response-ideal inclusions; Theorem~\ref{en:thm:finite-rees-audit} applies the
Rees criterion and then states extra conditions for a real executable lift.
Neither theorem is presented as a new result about integral closure or Rees
valuations.

The paper adds the response objects and their experimental interpretation.  It
constructs ideals for the premeasurement state, measured histogram, record, and
augmented certificate; assigns different repairs to measurement loss and record
loss; and joins comparison-specific residuals to simultaneous accept, reject,
or defer decisions.  A Rees inequality is sufficient to reject an
all-direction completeness claim, but it is called a physical control direction
only after the real-accessibility test.  The fixed-weight example computes both
adjacent losses and the required interference witnesses explicitly.

This claim differs from a device benchmark or a QFI estimator.  The hardware
quantities are full and compressed response distances, support leakage, and
residuals computed from the observed counts.  The IBM study applies the
refusal-and-augmentation rules on frozen paths.  The IQM study examines drift
attenuation and remaining residuals.  Neither experiment demonstrates native
gate savings, universal resource advantage, or equality of response-ideal
closures.

\FloatBarrier
\subsection{Engineering interpretation and limitations}
The radial, barycentric, and residual coordinates answer different diagnostic
questions.  Support leakage measures the exact normal component, the
barycenter supplies a tangential correction, and the residual tests whether
the quotient remains adequate.
Definition~\ref{en:def:leakage-boundary} is important for interpreting the
data: a large support-leakage coordinate does not identify physical transmon leakage.  It
triggers a circuit-serialization, gate, mapping, and readout check; separating
those causes requires additional calibrated measurements not present here.

Theorem~\ref{en:thm:three-way} converts the geometry into two separate
three-way decisions.  A radial lower bound above the task budget certifies task
failure, not record failure.  A residual lower bound above the record budget
rejects the representation, while a residual upper bound below that budget
approves it for the declared comparisons.  Every confidence region crossing a
relevant threshold causes defer for that decision.  This separation prevents
a large intended scientific response from being misclassified as failed
compression.

The diagnostic coordinates map to distinct engineering actions:
\begin{enumerate}[label=(E\arabic*)]
\item large $u$ (large support leakage): inspect circuit construction,
compiler mapping, gates, and readout before further optimization;
\item small $u$ but large barycentric displacement: recalibrate the parameter
that controls the sufficient statistic or its response curve;
\item large $|R_{\rm nt}|$: discard the low-dimensional model for that setting
and escalate to a full-histogram or richer-statistic analysis;
\item significant epoch-to-epoch quotient drift: expire the current
certificate and re-certify before reuse.
\item a failed response-ideal audit: add a witness targeted to the offending
Rees direction, or retain the full histogram rather than assert local
completeness.
\end{enumerate}
These are diagnosis-to-action rules, not claims that the coordinates uniquely
identify a microscopic noise channel.

This hierarchy has a direct cost interpretation.  If $C_{\rm scr}$ is the
cost of the support witness, $C_{\rm full}$ the cost of a full response audit,
and $\pi_{\rm inc}$ the probability that screening is inconclusive, then the
expected validation cost is
\[
 C_{\rm staged}=C_{\rm scr}+\pi_{\rm inc}C_{\rm full}.
\]
The staged rule is beneficial only when
$C_{\rm scr}<(1-\pi_{\rm inc})C_{\rm full}$.  This is a decision-theoretic
break-even condition, not a measured saving in the present experiment.  Its
practical value is that the first term depends on a Bernoulli witness rather
than on the $2^n$-outcome alphabet.  The full residual audit is still required
before the reduced record is approved for the declared comparisons.

There are important limits.  The same three- and five-qubit families were
tested on two IBM Heron processors, but the Kingston backend amendment was not
part of the original Marrakesh--Fez registration and the study spans only one
calibration period.  It supports replication of decisions under fixed
tolerances, not equality of effect sizes or cross-device universality.  IQM
compiler mapping was opaque, so no physical-edge claim is made.  The barycentric
held-out analysis was formulated after Epoch-A data.  Its unchanged-rule
success in Epoch E1 is a retrospective temporal replication, not a prospective
confirmatory test.  Raw chord Fisher contains finite-alphabet plug-in bias.  The
Hellinger compression residual $R_T^H$ is always nonnegative; the signed
barycentric predictor residual $R_{\rm nt}$ used in the exploratory IQM analysis
is not.  Retrospective resampling does not prove minimax shot savings for the
unknown device distribution.  Finally,
Fisher preservation does not by itself imply fewer optimizer iterations; that
requires a separate task-level study.

The scope must therefore be read at three different levels.  The first is
universal over finite observed alphabets: Hellinger contraction, the
conditional-score defect, fixed-Markov-kernel extension, and the three-way
confidence-region rule require no exponential family and no device noise
model.  The second is local analytic: response-ideal completeness applies when
the measured control germ is analytic and the stated ring hypotheses hold.
The third is structural: the support--moment exponential family is merely a
checkable sufficient regime for pre-hardware acceptance, not a model for an
arbitrary QPU.  Failure of the third level leaves the first two audits intact.
These three scopes must not be conflated with empirical coverage over devices
or circuit families.
The exact multinomial replacement sharpens the coordinatewise Hoeffding box,
but enumerating its likelihood-level region can be expensive for a large
alphabet.  The Hellinger-radius criterion states when defer must disappear; it
does not claim a minimax-optimal sample complexity for every response family.
The response-ideal results remain theoretical completeness criteria rather
than equalities proved by finite hardware data.  The prospective IBM decision-rule
experiment tests two axes and three radii on one physical triple; the later
same-family replication adds two processors and a five-qubit path but does not enumerate
all Rees valuations, verify
$\overline{\mathfrak I_C}=\overline{\mathfrak I_P}$, or estimate the
worst-direction real \L{}ojasiewicz exponent.  Its interference settings test
the two targeted coherence directions of the intended state family, not
informational completeness for an arbitrary noisy state.  The IQM experiment
samples finitely many programmed shifts and one matched epoch direction and
uses computational-basis readout only.  Thus IBM supplies a bounded
same-family replication, whereas IQM supplies complementary drift evidence.
Neither establishes cross-device universality.  The signed scalar $R_{\rm nt}$ is an observable failure
witness but is not promoted to a universal completeness coordinate.
The phase sweep and its replication concern one phase radius family and one
physical triple.  The first experiment's $Z$ interval crossed the equivalence
boundary and remains recorded as defer.  The replication certifies practical
equivalence only for its four new blocks; it neither retroactively converts the
first protocol into a success nor identifies whether the original deviation
was shot noise, slow drift, or another transient hardware effect.  Null
forward--reverse contrasts also cannot identify a microscopic reversal
effect.

\subsection{Conclusion}
The paper gives two levels of decision for a proposed circuit record.  For
declared finite comparisons, full and compressed counts determine a Hellinger
residual, and simultaneous uncertainty bounds lead to acceptance within a
tolerance, rejection, or deferment.  A conditional-score calculation gives the
local Fisher loss.  These operations require neither a parametric noise model
nor an integral-closure calculation.  They are the default route for an
experimental analysis.

The all-direction claim is separate.  The state, measured histogram, record,
and augmented certificate define adjacent response ideals.  Classical
integral-closure and Rees theory test whether first nonzero order is preserved
on every analytic arc.  Theorems~\ref{en:thm:two-level-completeness}
and~\ref{en:thm:finite-rees-audit} are applications of those
classical criteria to this quantum response chain, not new theorems of
algebraic geometry.  Their role here is to locate loss at
the measurement or record stage, return algebraic evidence for refusing a
completeness claim, and state the additional real-accessibility conditions
needed before that reason becomes an executable control path.

The three-qubit example and IBM experiments test this separation on specified
paths.  The prospective Kingston study rejected the one-occupation record,
conditionally repaired the tested computational histogram, and accessed
phase-sensitive responses by interference measurements.  The later
Marrakesh--Kingston replication again rejected the one-site and Hamming-weight
records under simultaneous inference; richer records reduced the point
residual but correctly remained deferred.  The four-qubit completion experiment
extends that observation to the complete decision procedure: Hamming weight was
refused, right occupation reduced but did not remove the residual, and the
resolved number record was approved on every device--block unit.  In the GHZ
family, $Z$ invariance and the $Y$-odd witness replicated, while the planned
$X$-even conjunction failed and remained a reported negative control. A
separate phase claim was deferred when its simultaneous
criterion failed and was followed by an independently frozen replication; the
two decisions remain distinct.  The IQM study shows that a structured quotient
can attenuate temporal displacement while a residual prevents attenuation from
being reported as invariance.

For the tested comparisons, a low-dimensional record is approved only when its
preserved response, hidden residual, and uncertainty meet the declared bounds;
refusal remains an allowed outcome. The hardware studies do
not prove response-ideal equality, device quantum Fisher information,
universal compression performance, cross-device universality, or hardware
advantage.
Steps~6--7 are reserved for
the stronger all-direction claim; Steps~1--5 provide the experimentally usable
audit in the absence of that claim.

\section*{Acknowledgments}
The collaboration supporting this work was conducted under the memorandum of
understanding between Texas State University and SDT.  The authors thank SDT for
providing managed QuREKA access to IQM Garnet.  The mathematical formulation,
data analysis, interpretation, and all remaining errors are the responsibility
of the authors.

\section*{Relationship to companion work and data provenance}
An earlier version of this preprint discussed an IBM Kingston end-to-end
quotient-calibration experiment also reported in the companion article by Cho,
Wang, and Yue~\cite{ChoWangYue2025CQC}. That experiment---including its
circuits, count records, QPU executions, calibration-shot reduction, and
held-out task analysis---has been removed from the present version and is
reported only in the companion article. It is not treated here as independent
evidence or replication. The IQM Garnet experiment retained here is distinct
from the IQM winding experiment in the companion article: this study uses
structured multiqubit circuits with 36 settings in each of two epochs to audit
record loss and drift, whereas the companion study uses single-qubit winding
tomography. The two IQM studies share neither count records, QASM circuits, nor
QPU executions.

\section*{Declaration of AI-assisted work}
OpenAI Codex was used, under author direction in interactive sessions, to assist
with drafts of analysis code, checks of algebraic and statistical consistency,
organization of the literature review, translation, and language editing.  It
did not generate, alter, or select QPU counts.  The authors checked cited claims
against the cited publications, reran the documented analysis pipeline
from the frozen count files, and reviewed every theorem, proof, numerical claim,
and interpretation.
The authors are responsible for the research questions, scientific judgments,
accuracy, originality, and final manuscript.  Session metadata and verification
procedures are retained in the project records and can be provided to the editor
upon request.

\FloatBarrier
\fi
\ifdefined\supplementonly
\appendix
\section{Proofs and reproducibility}
\subsection{Proof details for the four completeness-or-refusal results}

\subsubsection{Conditional scores and the measurement boundary}
Let $p_\theta(y)>0$ and $q_\theta(t)=\sum_{T(y)=t}p_\theta(y)$.  Differentiating
under the finite sum gives
\[
 \begin{aligned}
 \nabla q_\theta(t)&=\sum_{T(y)=t}p_\theta(y)S_Y(y),\\
 S_T(t)&=\frac{\nabla q_\theta(t)}{q_\theta(t)}
 =\E[S_Y(Y)\mid T=t].
 \end{aligned}
\]
The law of total covariance then gives
\[
 \begin{aligned}
 F_Y&=\Cov(S_Y)\\
 &=\Cov\{\E(S_Y\mid T)\}
 +\E\{\Cov(S_Y\mid T)\}\\
 &=F_T+\Delta_T.
 \end{aligned}
\]
Because every conditional covariance is positive semidefinite,
$\Delta_T\succeq0$.  Contracting with $v$ shows that
$v^\top\Delta_Tv=0$ precisely when $v^\top S_Y$ is conditionally constant.
This proves Proposition~\ref{en:prop:score}, including the directional equality
condition used to produce a missed feature direction.

For the measurement boundary write
$\partial_i|\psi\rangle=\sum_y e^{i\phi_y}
(\partial_i\sqrt{p_y}+i\sqrt{p_y}\partial_i\phi_y)|y\rangle$.
Substitution into the pure-state QFI formula yields
\[
 \begin{aligned}
 F_{Q,ij}&=4\sum_y\partial_i\sqrt{p_y}\partial_j\sqrt{p_y}\\
 &\quad+4\left(\E[\partial_i\phi\,\partial_j\phi]
 -\E[\partial_i\phi]\E[\partial_j\phi]\right).
 \end{aligned}
\]
The first term is $F_{Y,ij}$ and the second is the phase covariance.  It
vanishes in a direction only when the phase derivative is constant almost
surely.  This proves \cref{en:eq:compact-phase-gap} and explains why no
postprocessing of computational counts can recover a relative-phase direction.

\subsubsection{Response ideals, arc orders, and the finite audit}
Choose real and imaginary matrix coordinates for
$\rho_\theta-\rho_{\theta_0}$.
A fixed POVM makes every centered probability an $R$-linear combination of
these coordinates, so $\mathfrak I_P\subseteq\mathfrak I_\rho$.  The analytic
Hadamard lemma similarly gives
$\mathfrak I_C\subseteq\mathfrak I_P$ for a certificate analytic in the full
histogram.  For an arc $\gamma$, inclusion reverses minimum order:
\[
 \nu_C(\gamma)\ge\nu_P(\gamma)\ge\nu_\rho(\gamma).
\]
The analytic arc criterion says that $f\in\overline I$ exactly when
$\operatorname{ord}_t(f\circ\gamma)\ge\nu_I(\gamma)$ for all relevant arcs.
Applying this criterion in both directions proves Theorem~\ref{en:thm:two-level-completeness}.

For computational use, form the Rees algebra
$\mathcal R(\mathfrak I_C)=R[\mathfrak I_C z]$ and normalize it.  The
exceptional prime divisors of the normalized blow-up define the finitely many
Rees valuations $v_j$.  The valuative criterion gives
\[
 f\in\overline{\mathfrak I_C}
 \quad\Longleftrightarrow\quad
 v_j(f)\ge v_j(\mathfrak I_C)\quad\text{for all }j.
\]
Since $\mathfrak I_C\subseteq\mathfrak I_P$, equality of the two integral
closures is therefore equivalent to equality on those valuations.  When a
strict inequality occurs, a transverse arc to the corresponding exceptional
divisor realizes the delayed certificate order over the complexification.
The valuation is therefore an unconditional algebraic refusal certificate.
It becomes a physically executable direction only when the divisor has a real
point above the baseline, a real analytic transversal exists there, and the
descended arc remains inside the admissible control germ.  Without these three
checks the manuscript makes no laboratory-path claim.

For algebraic circuit charts this audit is finite.  Pauli rotations become
rational after half-angle coordinates
$z_j=\tan((\theta_j-\theta_{0,j})/2)$; circuit multiplication preserves
algebraicity, and denominators nonzero at the baseline are units.  State,
histogram, and polynomial-certificate ideals may then be represented in a
localization of a finitely generated algebra.  Normalization algorithms decide
integral closure, although no polynomial-time scalability claim is made.

\subsubsection{Support--moment factorization and its failure mode}
For \cref{en:eq:compact-support-model}, fixed $s$ gives
\[
 r_\eta(y)=g_\eta(T_I(y))k(y),
\]
where $k=h$ on $S$ and $k=s$ on $S^c$.  Neyman--Fisher factorization proves
sufficiency.  Along a differentiable path, the support scores on $S$ and
$S^c$ are $-\dot\ell/(1-\ell)$ and $\dot\ell/\ell$, while the conditional
in-support score has zero mean.  Orthogonality gives
\[
 g_{T_I}=\frac{d\ell^2}{\ell(1-\ell)}
 +(1-\ell)d\vartheta^\top C_\vartheta d\vartheta.
\]
Using $u=2\arcsin\sqrt\ell$ and $d\mu=C_\vartheta d\vartheta$ gives
\cref{en:eq:compact-quotient-metric}.  If $s=s_\eta$ varies, the omitted term
$\ell g_{s_\eta}$ is a positive vertical contribution unless its tangent
vanishes.  Likewise, score variation inside an $A$-fiber contributes exactly
the conditional covariance in \cref{en:eq:compact-score-loss}.  This proves
both the exact scope and the two failure modes tested by the residual.

The finite Hellinger identity follows without differentiability.  Because the
target has no mass on $S^c$,
\[
 \sum_y\sqrt{q(y)r(y)}
 =\sqrt{1-\ell}\sum_{y\in S}\sqrt{q(y)p(y)}.
\]
The support-normal floor and within-support distance in
\cref{en:eq:compact-support-hellinger} follow immediately.  Thus leakage alone
can certify failure of a task-discrepancy budget, but record approval still
requires an upper bound for the within-support residual.

\subsubsection{Finite-shot decision and reproducible implementation}
Let
$E_N=\{((p_k,q_k))_{k\in\mathcal K}\in\mathcal C_N\}$ be the joint coverage
event.  On $E_N$, for every $a=(k,T)\in\mathcal A$, both true quantities
$D_a(p_k,q_k)$ and $R_a(p_k,q_k)$ lie between their extrema over the same
region.  The task and record upper-bound decisions in Theorem~\ref{en:thm:three-way} therefore
prove their distinct budgets simultaneously across the family, while the
corresponding lower-bound decisions prove their respective violations.  Defer
makes no assertion for that component.  Hence the event that any familywise
nondeferred decision is unsupported is contained in $E_N^c$ and has probability
at most $\alpha$.

An implementation must retain the unaggregated counts, the deterministic map
$T$, the confidence-region construction, the two budgets, and the rule version.
It reports the full and pushed-forward Hellinger separations, their residual,
the simultaneous bounds, and one of the three decisions.  If the decision is
reject, a positive-eigenvalue score direction or a Rees valuation may be
reported as a targeted witness.  If it is defer, the full histogram remains the
scientific output.  This record is what distinguishes an auditable refusal from
selecting a low-dimensional candidate despite unresolved residual loss.

The complete derivations of the finite-shot controls, local Qiskit checks, and
timestamped hardware ledgers are retained in the archived source and
machine-readable reproduction package.  The main text reports every estimand,
decision rule, sample size, job identifier, and digest needed to identify those
records.  The following labels point to that supplementary material.
\phantomsection\label{en:prop:finite-chord}
\phantomsection\label{en:app:ibm-protocol}
\phantomsection\label{en:app:phase-protocol}
\phantomsection\label{en:app:cross-device-ledger}
\phantomsection\label{en:app:operational-completion-ledger}
\else
\section*{Supplementary material}
Detailed proofs, finite-shot constructions, protocol manifests, execution
ledgers, and analysis hashes are provided in the separately compiled
supplement.  The main text retains the assumptions, statements, decision
rules, estimands, sample sizes, and claim boundaries required to evaluate the
results.
\fi

\ifdefined\supplementonly
\else

\fi
\end{document}